\documentclass[review,onefignum,onetabnum]{siamart171218}
\usepackage{euscript,amsmath,amssymb,graphicx,color,setspace,cite}
\usepackage{epstopdf}
\usepackage[top=1in, bottom=1in, left=1in, right=1in]{geometry}
\usepackage{marvosym}
\usepackage{hyperref}

\newcommand{\R}{{\mathbb R}}

\newcommand{\LLR}{{\rm LLR}}

\newcommand{\mc}{\mathcal }

\newcommand{\ep}{\epsilon}

\newcommand{\bi}{\begin{itemize}}
\newcommand{\ei}{\end{itemize}}
\newcommand{\be}{\begin{enumerate}}
\newcommand{\ee}{\end{enumerate}}

\newcommand{\bv}[2]{\left( \begin{array}{c} #1 \\ #2 \end{array} \right)}

\newcommand{\hm}{{H^-}}
\newcommand{\hp}{{H^+}}
\newcommand{\hpm}{{H^{\pm}}}
\newcommand{\defi}[1]{\emph{#1}}
\newcommand{\dec}{{\mathrm{Soc}}}
\newcommand{\ob}{{\mathrm{Priv}}}
\newcommand{\priv}{{\rm Priv}}
\newcommand{\soc}{{\rm Soc}}
\newcommand{\vb}{\mathbf v}

\newcommand{\obb}[2]{{\ob^{(#1)}_{#2}}}
\newcommand{\thp}{{\theta_+}}
\newcommand{\thm}{{\theta_-}}

\definecolor{darkgreen}{rgb}{0,0.4,0}

\definecolor{psimon}{rgb}{0.5,0.3,0.3}

\title{Bayesian evidence accumulation on social networks}
\author{Bhargav Karamched\thanks{Department of Mathematics, University of Houston, Houston TX 77204 ({\tt bhargav@math.uh.edu})}, \and
Simon Stolarczyk\thanks{Department of Mathematics, University of Houston, Houston TX 77204 ({\tt simon.stolarczyk@gmail.com})}, \and
Zachary P. Kilpatrick\thanks{Department of Applied Mathematics, University of Colorado, Boulder CO 80309 ({\tt zpkilpat@colorado.edu})}, \and
Kre\v{s}imir Josi\'{c}\thanks{Departments of Mathematics, and Biology and Biochemistry, University of Houston, Houston TX 77204, Department of BioSciences, Rice University, Houston TX 77005 ({\tt josic@math.uh.edu}); Bhargav Karamched and Simon Stolarczyk are co-first authors; Zachary Kilpatrick and Kre\v{s}imir Josi\'{c} are co-last authors.}}

\begin{document}

\nolinenumbers

\maketitle

\begin{abstract}
To make decisions we are guided by the evidence we collect and the opinions of friends and neighbors.
How do we combine our private beliefs with information we obtain from our  social network? To understand the strategies humans use to do so, it is useful to compare them to observers that optimally integrate all evidence. Here we derive network models of rational (Bayes optimal) agents who accumulate private measurements and observe decisions of their neighbors to make an irreversible choice between two options. The resulting information exchange dynamics has
interesting properties: When decision thresholds are asymmetric, the absence of a decision can be increasingly informative over time. In a recurrent network of two agents, an absence of a decision can lead to a sequence of belief updates akin to those in the literature on \emph{common knowledge}.  On the other hand, in larger networks a single decision can trigger a cascade of agreements and disagreements that depend on the private information agents have gathered. Our approach provides a bridge between social decision making models in the economics literature, which largely ignore the  temporal dynamics of decisions, and the single-observer evidence accumulator models used widely in neuroscience and psychology.
\end{abstract}


\section{Introduction}

Understanding how organisms use sensory and social information to make decisions is of fundamental interest in biology, sociology, and economics~\cite{Frith2012,Couzin2009,Edwards54}.  Psychologists and neuroscientists have developed a variety of experimental approaches to probe how humans and other animals make choices.  Popular are variants of the two-alternative forced choice task (2AFC) where an observer is asked to decide between two options based on information obtained from one or more noisy observations~\cite{Bogacz2006,Saaty90,Gold2007,Ratcliff2008}.  The 2AFC task has motivated several mathematical models that successfully explain how humans use sequentially-presented evidence to make decisions~\cite{Busemeyer1993, Usher2001,degroot2005optimal}.  

Most such evidence accumulation models take the form of drift-diffusion stochastic processes that describe the information gathered by  lone observers~\cite{Bogacz2006}.  However, humans and many other animals do not live in isolation and may consider their neighbors' behavior as they make decisions.  Animals watch each other as they forage~\cite{Miller2013,Krajbich2015}.  Stock traders, while not privy to all of their competitor's information,  can observe each other's decisions.  To make the best choices many  biological agents thus also take into account the observed choices of others~\cite{Raafat2009,Moussaid2010,Park17}. 

{Key to understanding the collective behavior of social organisms is uncovering how they gather and exchange information to make decisions~\cite{Couzin2009,Frith2012}. Mathematical models allow us to quantify how different evidence-accumulation strategies impact experimentally observable characteristics of the  decision-making process such as the speed and accuracy of choices~\cite{bogacz2010} and the level of agreement among members of a collective~\cite{conradt2005}.  Such models can thus lead the way towards understanding the decisions and disagreements that emerge in social groups.}

Here we address the question of how idealized rational agents in a social network should combine a sequence of private measurements with observed decisions of other individuals to choose between two options.  We refer to the information an agent receives from its neighbors as \defi{social information} and information available only to the agent as \defi{private information}.  As agents  do not share their private information with others directly, they only reveal their beliefs through their choices. These choices are based on private and social information and thus reveal something about the total evidence an agent has collected.

We assume that  private measurements and, as a consequence, observed choices can improve the odds of making the right decision. However neither type of information affords certainty about which  choice is correct. We take a probabilistic (Bayesian) approach to describe the behavior of rational agents who make \emph{immutable} decisions once they have accrued sufficient evidence. {Previous models of this type are not normative, as the belief update in response to an observed decision is a parameter~\cite{caginalp2017decision}. Hence such models cannot describe the interplay between the decisions, non-decisions, and the belief updates we discuss here.}


There are two reasons for assuming that agents only share their decision: first, many organisms  communicate their decisions but not the evidence they used to reach them. For example, herding animals in motion can only communicate their chosen direction of movement~\cite{Couzin2005,Nagy2010}. Foraging animals may communicate their preferred feeding locations~\cite{Seeley1991,Franks2006} but not the evidence or process they used to decide. Traders can see their competitor's choice to buy or sell a stock but may not have access to the evidence that {leads} to these actions.  Second, if agents communicate all information they gathered to their neighbors, the problem is mathematically trivial as every agent 
obtains all evidence provided to the entire network.

The behavior and performance of rational agents in a social network can depend sensitively on how information is shared~\cite{Acemoglu2011,mossel2014}. In some cases rational agents perform better in isolation than as part of a social network, even when they use all available information to decide on the most likely of several options.  This can happen when social information from a small subset of agents dominates the decisions of the collective, as in classical examples of \emph{herding behavior}~\cite{banerjee1992,Gale2003}.  {In recurrent networks, on the other hand, agents can repeatedly share their beliefs with one another and asymptotically come to agree on the most probable choice given the evidence each obtains~\cite{aumann1976,geanakoplos1994,Geanakoplos1982,mossel2014}.  Often all agents are  able to assimilate all private information in the network just by observing the preferences of their neighbors leading to \emph{asymptotic learning}.} 

{
We show that such asymptotic learning typically does not occur in our model, as the  decisions are immutable.  We describe what information
rational agents can obtain from observing the decision states of their neighbors.   Interestingly when decision thresholds are asymmetric, the absence of a decision in recurrent
networks can lead to a recursive belief update akin to
that addressed in the literature on  {\em{common knowledge}}~\cite{aumann1976,geanakoplos1994,milgrom1982}.  We also show how a rational agent who only observes a portion of the network must  marginalize over the decision states of all unobserved agents to correctly integrate an observed decision. }

{
Social evidence exchange in larger networks becomes tractable when all agents can observe each other.  We show that
in such networks the first decision can lead to a wave of agreements. However, the agents who remain undecided
can obtain further information by counting how many others have made a decisions and how many remain on the fence.
We show that this process can be self-correcting, and  lead to quicker and more accurate decisions compared to lone agents.}

\section{Definitions and setup}\hspace{-2.5mm}\footnote{For the convenience of the reader a table of notation and symbols is available at  \url{https://github.com/Bargo727/NetworkDecisions.git}, along with the code used to generate all figures.}
We consider a model of a social network in which all agents are deciding between two options or hypotheses.
To make a choice, agents gather both private (${\rm Priv}$) and social (${\rm Soc}$) information over time. We assume that private information comes from a sequence of noisy observations (measurements). In addition we assume that agents also gather social information by continuously observing each other's decisions. For instance, while foraging, animals make private observations but also observe each other's decisions\cite{Couzin05,Miller2013}. When deciding on whether or not to purchase an item or which of two candidates to vote for people will rely on their own research but are also influenced by the decisions of their friends and acquaintances~\cite{Acemoglu2011} as well as opinions on social networking sites~\cite{Watts2002}.
In many of these situations agents do not share all information they gathered directly but only observe their neighbors' choices or the absence of such choices (\emph{e.g.} not purchasing an item, or not going to the polls). 

More precisely, we consider a set of agents who accumulate evidence to decide between two states, or hypotheses, $\hp$ or $\hm$.
Each agent is rational (Bayesian): they compute the probability that either hypothesis holds based on all evidence they accrued. The agents make a decision once the conditional probability of one hypothesis, given all the accumulated observations, crosses a predetermined threshold~\cite{Wald1948,Bogacz2006}.  For simplicity, we  assume that all agents in the network are identical, but discuss how this condition can be relaxed.


{\bf Evidence accumulation for a single agent.} The problem of a single agent accumulating private evidence to decide between two options has been thoroughly studied~\cite{Wald1948,Usher2001,Gold2007,Bogacz2006,Ratcliff2008,radillo17,veliz16}. In the simplest setting  an agent makes a sequence of noisy observations, $\xi_{1:t}$ with $\xi_i \in \Xi, i \in \{1, \ldots, t\}$, and $\Xi \subset \R$ finite. The observations, $\xi_i,$ are independent and  identically distributed, conditioned on the true state $H \in \{\hp,\hm\}$,
\[
P(\xi_{1:t} | H^{\pm} ) = \prod_{i=1}^t P(\xi_i | H^{\pm} ) =   \prod_{i=1}^t f_{\pm} ( \xi_i) , 
\]
where the probability of each measurement is given by the probability mass functions $f_{\pm}(\xi) := P(\xi | H^{\pm})$. Observations, $\xi_i,$ are drawn from the same set $\Xi$ in either state $H^{\pm}$, and the two states are distinguished by the differences in probabilities of making certain measurements.   

 To compute  $P(H^{\pm}| \xi_{1:t}) $ the agent uses Bayes' Rule:  for simplicity, we assume that the agent knows the measurement distributions $f_{\pm}(\xi)$ and uses a flat prior,  $P(\hp) = P(\hm) = 1/2$. Thus,  the log  likelihood ratio (LLR) of the two states at time $t$ is
\begin{align} \label{E:one_agent}
y_t & := \log \left( \frac{P( \hp | \xi_{1:t} )}{P( \hm | \xi_{1:t} )} \right) = \sum_{s = 1}^{t} \LLR (  \xi_s) = y_{t - 1} + \LLR (  \xi_t),
\end{align}
{where we define $\LLR(\cdot) = \log \frac{P(  \cdot  | \hp ) }{P( \cdot | \hm )} $; such log likelihood ratios will be used extensively in what follows.} Also note that $y_0 = 0$, since both states are equally likely \emph{a priori}.   We also refer to the LLR as the \emph{belief} of the agent. 

An ideal agent continues making observations while $\thm < y_t < \thp $, and makes a decision after acquiring sufficient evidence,  choosing $\hp$ ($\hm$) once $y_t \geq \thp$ ($y_t \leq \thm$). We assume $\thm < 0 < \thp$.

\section{Multiple agents} 
\label{sec:uni}

In social networks not all agents will communicate with each other due to limitations in bandwidth or physical distances~\cite{Lazer2009,Newman2003,Yan2006}. To model such exchanges, we identify agents with a set of vertices, $V = \{1,...,N\}$, and communication between agents with a set of directed edges, $E,$ between these vertices~\cite{Goyal2012}. Information flows along the direction of the arrow (Fig.~\ref{fig:uni}a), so agents observe neighbors in the direction opposite to that of the arrows.

As in the case of a single observer, we assume that each agent, $i$,  makes a sequence of noisy, identically-distributed measurements, $\xi^{(i)}_{1:t},$  from a state--dependent distribution, $f_{\pm}(\xi^{(i)}_j) = P(\xi^{(i)}_j | H^{\pm})$.  We assume that the observations are independent in time and between agents, conditioned on the state, $P(\xi^{(i_1)}_{1:t}, \ldots, \xi^{(i_k)}_{1:t} |H) =   \Pi_{l = 1}^k  \Pi_{j = 1}^t  P(\xi^{(i_l)}_{j}|H)$ for any sequence of measurements, $\xi^{(i)}_{1:t}$, and set of agents, $i_1, \ldots, i_k$, in the network. This  conditional independence of incoming evidence  simplifies calculations, but is unlikely to hold in practice.  However, humans often treat redundant information as  uncorrelated, thus making the same independence assumption as we do in our model~\cite{enke2013correlation,levy2015correlation} {(See Section~\ref{sec:discussion}).}

Each agent gathers social evidence by observing whether its neighbors have made a decision and what that decision is.  Each agent thus gathers private and social evidence and makes a decision when its belief (LLR) about the two states crosses one of the thresholds, $\thm<0<\thp$. Importantly, \emph{once an agent has made a decision, it cannot change it.}   The absence of a decision thus communicates that an agent has not gathered sufficient evidence to make a choice and hence that this agent's belief is still in the interval $(\thm,\thp)$.

For simplicity, we assume  that the measurement distributions, $f_{\pm},$ and thresholds, $\theta_{\pm},$ are identical across agents. The theory is similar if  agents have different, but known, measurement distributions.  The assumption that the  distributions, $f_{\pm}(\xi),$ are discrete simplifies some convergence arguments. However, most evidence accumulation models take the form of continuous, drift-diffusion stochastic processes~\cite{Bogacz2006}.  These models  take the form of stochastic differential equations (SDEs) and approximate the discrete model well when many observations are required to reach a decision~\cite{veliz16}. Such SDEs have been remarkably successful in describing the responses of humans and other animals under a variety of conditions~\cite{Ratcliff:1998}.  

{\bf Evidence accumulation with two agents.} To illustrate how an agent integrates private and social information to reach a decision, we use the example network shown in Fig.~\ref{fig:uni}a.  

\begin{figure}
\centering 
\includegraphics[width = 12cm]{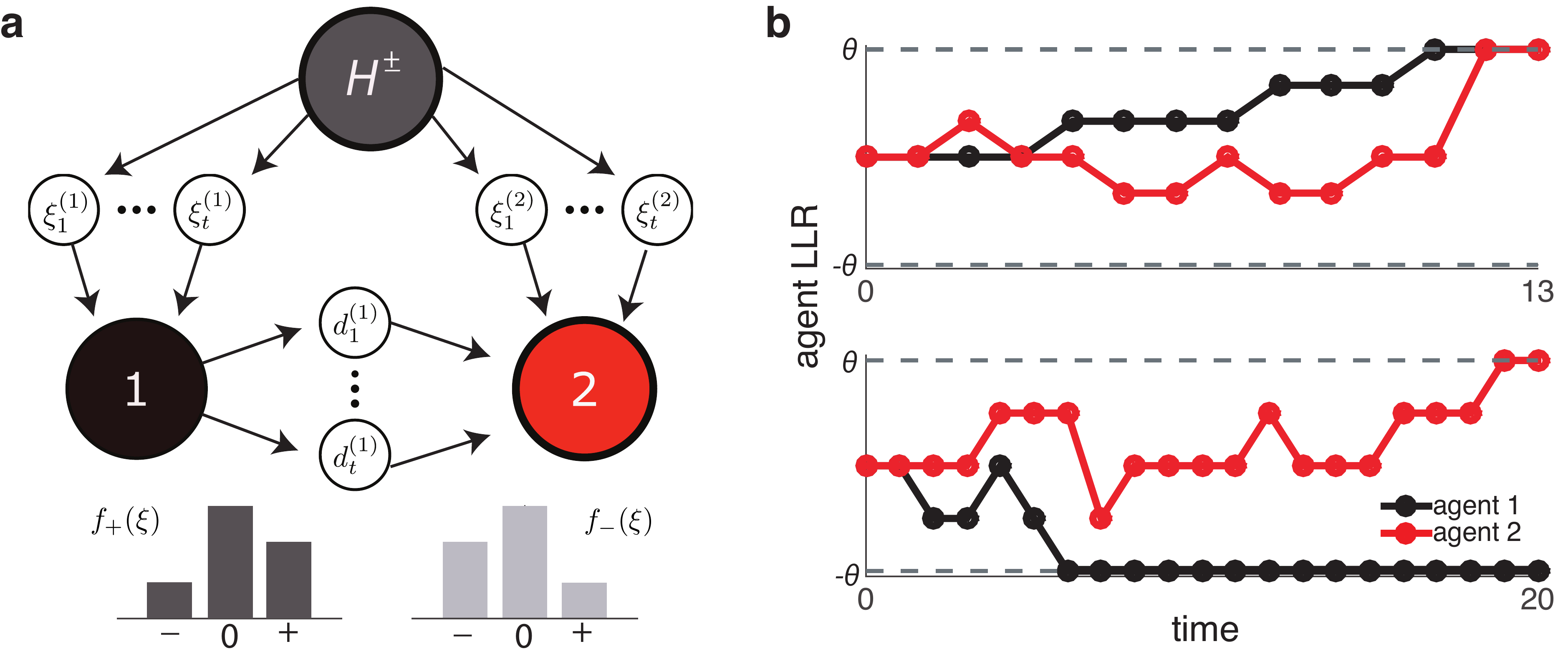}
\caption{A pair of unidirectionally coupled agents deciding between the states $H^+$ and $H^-$ on the basis of private and social evidence.  (a) Schematic of the information flow in the network. Agent 1 accumulates their own observations, $\xi^{(1)}_{1:t},$ which result in a sequence of decision states, $d^{(1)}_{1:t},$ that is observed by agent 2. In addition, agent 2 gathers its own observations, $\xi^{(2)}_{1:t},$ to make a decision.  (b) Sample trajectories for the beliefs (LLRs) of the agents.  Decisions are made when an agent's belief  crosses a threshold, $\theta_{\pm} = \pm \theta$ in this case. A decision of agent 1 leads to a  jump in the belief of agent 2. {Here, $H = H^+$}.\vspace{-.2cm}}
\label{fig:uni}
\end{figure}

Let $I_t^{(i)}$ be the total information available to agent $i$ at time $t$. The information available to agent 1,  $I_t^{(1)},$ consists only of private observations. However, agent 2 makes private observations and obtains social information from agent 1, which jointly constitute $I_t^{(2)}$.   Agents base their decisions on the computed LLR, or belief, $y_t^{(i)}= \log \left[ P(\hp | I_t^{(i)}) / P(\hm |  I_t^{(i)}) \right] $, and agent $i$ makes a choice at  time $T^{(i)}$ at which $y_t^{(i)}$ crosses  threshold $\theta_+$ or $\theta_-$.

Since $I_t^{(1)} = \xi_{1:t}^{(1)}$, the belief of agent 1 is described by Eq.~\eqref{E:one_agent}.  At each time $t$, agent 2 observes the  resulting \defi{decision state}, $d_t^{(1)},$ of agent 1 (Fig. \ref{fig:uni}a), where
\begin{equation} \label{E:decision-state}
d_t^{(1)} = \begin{cases} 
	-1, & y_t^{(1)} \leq  \thm, \\ 
	0, & y_t^{(1)} \in (\thm , \thp), \\ 
    1, & y_t^{(1)} \geq \thp.
    \end{cases}
\end{equation}
The decision state captures whether agent 1 made a decision by time $t$ and, if so, what that decision was.  

 Assuming that the two agents make private observations synchronously, the information available to agent 2 at time $t$ consists of its private observations and observed decision states of agent 1,  $I_t^{(2)} = (\xi_{1:t}^{(2)}, d_{1:t}^{(1)})$.  {Importantly, we assume that an observer makes a private observation and then integrates all available social information {\em{before}} making its next private observation.}  Hence, evidence accumulation at time $t$ consists of two steps:  Agents first updates their belief in response to a private observation. They then observe the decision states of other visible agents and update their belief again.  {We will see that in recurrent networks the exchange of social information  may occur over several steps (see Sections~\ref{sec:bidsymm},\ref{sec:3cliq},\ref{sec:cliq}).  
} 

 {
We let $\soc_t^{(i)}$ be the social information that can be obtained by observing the decisions of agent $i$ up to time $t$ so that
\begin{equation} \label{E:soc}
 \soc_t^{(1)} = \LLR ( d_{1:t}^{(1)} )
\end{equation}
}The stochastic process $\priv_{1:t}$ is defined in the same way for both agents and is equivalent to that of a lone observer given in Eq.~\eqref{E:one_agent}.

{
Belief updates rely on the assumption that the private observations of the two agents are conditionally independent and a straightforward application of Bayes' Theorem. In the present example each private observation leads agent 2 to update its belief to $y_{t,0}^{(2)}$ by adding Priv$^{(2)}_{t}$.  After next observing the decision state of agent 1, the belief of agent 2 is a sum of the LLR corresponding to private observations (Priv$^{(2)}_{1:t}$) and the LLR corresponding to the social information (Soc$_{t}^{(1)}$):
\begin{align} 
 y_{t,1}^{(2)} & = \log \frac{ P(\hp | \xi_{1:t}^{(2)}, d_{1:t}^{(1)})}{P(\hm |  \xi_{1:t}^{(2)}, d_{1:t}^{(1)})} =
 \underbrace{\sum_{l = 1}^t \LLR (\xi_l^{(2)})}_{\priv^{(2)}_{1:t}} + 
\underbrace{\LLR ( d_{1:t}^{(1)})}_{\soc_t^{(1)} }.
\label{E:decomposition}
\end{align}
As decisions are immutable, the only information about agent 1's decision that agent 2 can use after a decision at time $T$ is that  agent 1's decision state switched from $d_{T-1}^{(1)} = 0$ to $d_{T}^{(1)} = \pm 1$ at $t = T$.}

\section{Social Information}
\label{sec:unidir}

We next ask how much evidence is provided by observing the decision state of other agents in the directed network in Fig.~\ref{fig:uni}a and illustrate how non-decisions, $d_t^{(1)} = 0$, can be informative. 

{\bf Decision Evidence} Suppose agent 1 chooses $\hp$ at time $T$ so that the belief of agent 2 at time $t \geq T$ is $y_{t,1}^{(2)} = \priv_{1:t}^{(2)} + \soc_t^{(1)} = 
\priv_{1:t}^{(2)} + \soc_{T}^{(1)}$.
If  $d_{T}^{(1)} = 1$, then agent 2 knows that  $y_{T,0}^{(1)} \geq \thp$. Agent 1 has reached its decision based on private observations; hence, none of the social information obtained by agent 2 is redundant. 

The belief of agent 1 at the time of the decision, $T,$ could exceed threshold.  However, this excess is small if the evidence obtained from each measurement is small.  The following proposition shows that, after observing a decision, agent 2 updates its belief by an amount close to $\thp$ (Fig. \ref{fig:uni}b).  

\begin{proposition} 
\label{prop:bump}
{If agent 1 decides at time $T$, and $d_{T}^{(1)} = 1,$ then
\[
\theta_{+}  \leq \soc_{T} < \theta_{+} + \ep_{T}^{+}, \ \ \ \ \text{where} \ \ \ \
 0<\ep_{T}^{+} \leq  \sup_{\xi_{1:T} \in {\mc C}_+(T)}\LLR (\xi_{T}),
 \]
and ${\mc C}_+(T)$ is the set of all chains of observations, $\xi_{1:T},$  that trigger a decision in favor of  $H^{+}$ precisely at time $T$.  An analogous result holds when $d_{T}^{(1)} = -1$.}

\end{proposition}

\begin{proof}
\label{A:prop41}
Assume $d_{T} = 1$, and note that $d_{{T}-1}^{(1)} = 0$ and $d_{T}^{(1)} = 1$ imply $\sum_{t=1}^{T} \LLR (\xi_t^{(1)}) \geq \theta_+$, so
\begin{align}
\label{eq:chainprob}
 e^{\thp} \prod_{t = 1}^{{T}} f_-(\xi_t^{(1)} ) \leq \prod_{t = 1}^{{T}} f_+(\xi_t^{(1)} ).
\end{align}
We define the set of all chains of observations, $\xi_{1:T}$, that lead to a decision precisely at $T$: ${\mc C}_+(T) = \{ \xi_{1:T}^{(1)}~|~\obb 1{1:t} \in(\thm, \thp),~1 \leq t < T;~\obb 1{1:T} \geq \thp   \}$. Thus, $P( d_{{T}-1}^{(1)} = 0, d_{T}^{(1)} = 1 | \hp) = \sum_{{\mc C}_+(T)} f_+( \xi_{1:{T}}^{(1)} )$, where
$f_+(\xi_{1:T}^{(1)}) : = \prod_{t = 1}^{T} f_+(\xi_t^{(1)})$. Together with inequality~\eqref{eq:chainprob} this yields
\begin{align*}
P( d_{{T}-1}^{(1)} = 0, d_{T}^{(1)} = 1 | \hp ) \geq  e^{\thp} \sum_{{\mc C}_+(T)}  f_-( \xi_{1:{T}}^{(1)} ) = e^{\thp} P( d_{{T}-1}^{(1)} = 0, d_{T}^{(1)} = 1 | \hm);
 \end{align*}
hence, $\LLR ( d_{{T}-1}^{(1)} = 0, d_{T}^{(1)} = 1 ) \geq \thp$. 
To obtain the upper bound on the social information inferred by agent~2 from the decision of agent 1 ($ \soc_{T_1}^{(1)}$), we must bound the possible private information available to agent 1 ($ \obb 1 {1:T}$). Note for any $\xi_{1:T}^{(1)} \in {\mc C}_{+}(T)$,
 \[
 \obb 1 {1:T} = \obb 1 {1:T-1} + \LLR (\xi^{(1)}_{T}) < \theta_+ + \LLR (\xi^{(1)}_{T}).
 \]
which implies
 \[
 \ep_{T}^+ \leq \sup_{\xi_{1:T} \in {\mc C}_+(T)}\LLR (\xi_{T}),
\]
by positivity of $\LLR (\xi_{T})$, since $\xi_{T}$ triggers a decision, and $\LLR (\xi_T) \leq \sup_{\xi_{1:T} \in {\mc C}_+(T)}\LLR (\xi_{T})$.
A similar argument can be used to derive bounds on $\soc _{T}^{(1)}$ for $d_{T}^{(1)} = -1$.
\end{proof}

Note that the evidence accumulation process will be analogous if the agents' thresholds differ.  In particular, if agent 1's thresholds are $\theta_{\pm}^1$ and agent 2's are $\theta_{\pm}^{2}$, then if agent 2 observes a decision by agent 1, it will update its LLR by approximately $\theta_{\pm}^{1}$, depending on the decision of agent 1.

In the following we will frequently assume that the two distributions $f_+(\xi)$ and $f_-(\xi)$ are close, and hence $\LLR(\xi)$ is small for any $\xi$.  Proposition~\ref{prop:bump} then implies that $\soc_{T} \approx \theta_{+}$. This is a common approximation in the literature on sequential probability ratio tests~\cite{wald1945}.

{\bf Social information prior to a decision.}  To understand why the absence of a decision can provide information, consider the case when $\theta_+ \neq -\theta_-$.  As one of the thresholds is closer to naught, in general each agent is more likely to choose one option over another by some time $t$. The absence of the more likely choice therefore reveals to an observer that the agent has gathered evidence favoring the alternative.   We show this explicitly in the following.

\begin{definition}
\label{def:symm}
The  measurement distributions $P(\xi | \hp ) = f_+ (\xi)$ and $P(\xi | \hm ) = f_- (\xi)$ are \defi{symmetric} if there exists an involution $\Phi:  \Xi \to \Xi$, i.e. $\Phi = \Phi^{-1}$, with $\Phi(\xi) \neq Id(\xi)$ such that  $f_+ (\xi) = f_- (\Phi(\xi))$ for every $\xi \in \Xi$.  When $\thp = - \thm$ we say that the \defi{thresholds are symmetric}. 
\end{definition}

It is frequently assumed, and experiments are frequently designed, so that threshold and measurement distributions are symmetric~\cite{shadlen1996}.  In much of decision-making literature, for example, it is assumed that $\Xi = -\Xi$ and $\Phi(\xi) = -\xi$.  {In the following we let $\Phi(\xi) = -\xi$ when discussing distribution symmetry.} However, there are a number of interesting consequences when introducing asymmetries into the reward or measurement distributions~\cite{Balci10}, which suggest subjects adapt their priors or decision thresholds~\cite{hanks2011}. In examples we will assume that agents use asymmetric decisions thresholds ($\theta_+ \neq - \theta_-$) due to a known asymmetry in the 2AFC task.  If the two thresholds differ then an agent more readily adopts one of the two options.  When the measurement distributions are asymmetric then the agent can more easily obtain information about one of the two options. 
  In such a situation, non-decision on the part of an agent will provide information to those agents observing it.


We call the social information agent 2 gathers before observing a choice by agent 1, the \defi{non-decision evidence}. This social information is determined by the decomposition in Eq.~\eqref{E:decomposition}, and Eq.~\eqref{E:soc}. The \defi{survival probability} of the 
stochastically evolving belief is given by
\[
S_\pm (t) = P( d_t^{(1)} = 0 | H^\pm ) = P( y^{(1)}_{s,0} \in (\thm,\thp), 0 \leq s \leq t | H^\pm ).
\]
Then the social information provided by the absence of the decision by time $t$ is,
\begin{equation}
\label{social}
\dec_t^{(1)} = \log \frac{S_+(t)}{S_-(t)} = \LLR (T > t).
\end{equation}
Note that if, for example, $S_-(t) \leq S_+(t)$ for all $t \geq 0$, then $H^-$ decisions more often occur sooner than $H^+$ decisions, and $\log [S_+(t)/S_-(t)] \geq 0$. Thus, observing that an agent has not made a choice by time $t$ provides evidence in favor of the choice that requires more time to make.



{When the threshold and measurement distributions are symmetric, then every sequence of observations
$\xi_{1:t}^{(1)}$ favoring one decision has a matching sequence,  $\Phi(\xi_{1}^{(1)}), \ldots \Phi(\xi_{t}^{(1)})$, 
providing equal evidence for the opposite 
decision since 
\[
\sum_{l = 1}^t \LLR ( \xi_{l}^{(1)} ) =
- \sum_{l = 1}^t \LLR ( \Phi( \xi_{l}^{(1)} )) 
\]
By symmetry, both sequences of observations are equally likely, implying $S_+(t) = S_-(t)$.
We have therefore shown that:}

\begin{proposition} \label{prop:symmetric}
If the measurement distributions and thresholds are symmetric then $\dec_t^{(1)} = 0$ when $d_t^{(1)} = 0$. 
\end{proposition}


{We show below that when measurement distributions and thresholds are symmetric, non-decisions are uninformative in any network {until the first decision is made}. When thresholds and measurement distributions are not symmetric the absence of a decision of agent 1 can provide evidence for one of the choices.  In particular, \emph{social information from a non-decision is deterministic}, which we illustrate in the following example. }


\subsubsection*{Example: Belief as a Biased Random Walk}
\label{sec:4state}
In the following example we chose the measurement distributions,  $f_{\pm}( \xi ),$ so that the belief increments due to private measurements satisfy $\LLR ( \xi ) \in \{ +1, 0, -1 \}$ with respective probabilities $\{p,s,q \}$ ($\{q,s,p \}$) when $H = H^+$ ($H = H^-$).  {In particular, we take $p/q = e$.}  This ensures that a decision based solely on private evidence results from a belief (LLR) exactly at integer thresholds\footnote{In the remainder of this work we use the same measurement distributions in all examples.  We always choose $p,q,s$ so that decisions based solely on private evidence result in a LLR exactly at threshold.}. Here we take $\theta_+ = 2$ and $\theta_- = -1$. Beliefs then evolve as  biased random walks on the integer lattice, as demonstrated below.

 Agent 1 makes only private observations, while agent 2 makes private observations and observes agent 1.  Their beliefs are described by Eq.~\eqref{E:one_agent} and Eq.~\eqref{E:decomposition}, respectively.  Realizations of these processes are shown in Fig.~\ref{fig:fourStateExample}. The social evidence obtained by agent 2 is independent of the particular realization until agent 1 makes a decision, whereas private information is realization-dependent. When thresholds are small, an expression for social evidence can be obtained explicitly.

\begin{figure}[t!]
\centering
\includegraphics[width = 12cm]{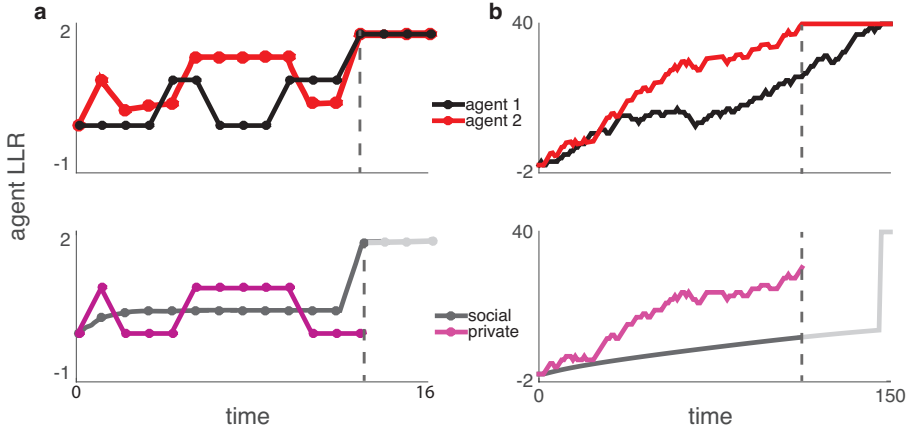}
\caption{An example of the belief evolution in a two-agent, unidirectional network shown in Fig.~\ref{fig:uni}.  (a) The belief of agent 1 is a random walk on the integers. For boundaries at $\thm = -1$ and $\thp = 2$, initially two observations favoring $\hp$ are required for decision $d_2^{(1)} = 1$.   The belief of agent 2 is the sum of stochastically evolving private evidence, and deterministic social evidence until observing a choice. The belief of agent 2 is incremented by $\theta_{\pm}$ when the agent observes a decision. The social information equals the thresholds in this case as by design agent 1 makes a decision when its belief equals $\theta_{\pm}$.
(b) The same processes with $\thm = -2$ and $\thp = 40$.  {Here, $H = H^+$.}
Code to produce all the figures in the manuscript is available at \url{https://github.com/Bargo727/NetworkDecisions.git}. }
\label{fig:fourStateExample}
\end{figure}

Standard first-step analysis~\cite{taylor2014introduction} yields non-decision social evidence in our case $\theta_+ = 2$ and $\theta_- = -1$. Here we have four belief states, $y_t^{(1)} \in \{ -1,0,1,2\}$, where $y_t^{(1)} = +2 =  \theta_+$ and $y_t^{(1)} = -1 = \theta_-$ are absorbing boundaries. Let $P_+^l(t)  :=  P(y^{(1)}_t = l | \hp)$. Then, for $\mathbf{P}_+(t) : = (P_+^{-1}(t), P_+^0(t), P_+^1(t), P_+^2(t))^T$, the probabilities are updated according to
\begin{align*}
\mathbf{P}_+ (t+1) = M(p,q) \mathbf{P}_+ (t), \hspace{5mm} M(p,q) := \left( \begin{array}{cccc} 1 & q & 0 & 0 \\ 0 & s & q & 0 \\ 0 & p & s & 0 \\ 0 & 0 & p & 1 \end{array} \right)
\end{align*}
with initial condition $P_+^0(0) = 1$ and $s = 1-p-q$. The rates $p$ and $q$ switch places in the $P_-^l : = P(y_t^{(1)} = l | H^-)$ case.

We solve explicitly for the evolution of the probability of the unabsorbed states $\vb_+(t) = (P_+^0(t), P_+^1(t))^T$ using the eigenpairs of the submatrix:
\[
\vb_+ (t) = \frac{\lambda_+^t}{2}\bv{1}{\sqrt{p/q}} -\frac{\lambda_-^t}{2}\bv{-1}{\sqrt{p/q}},
\]
where $\lambda_{\pm} := s \pm \sqrt{pq}$, and for $\vb_-(t) = (P_-^0(t), P_-^1(t))^T$ we have
\[
\vb_- (t) = \frac{\lambda_+^t}{2}\bv{1}{\sqrt{q/p}} -\frac{\lambda_-^t}{2}\bv{-1}{\sqrt{q/p}}.
\]
Thus for $d_t^{(1)}=0$,
\begin{align*}
\dec_t^{(2)} &
= \log \frac{ (\lambda_+^t + \lambda^t_-) + \sqrt{p/q} ( \lambda_+^t - \lambda_-^t)}
{(\lambda_+^t + \lambda^t_-) + \sqrt{q/p} ( \lambda_+^t - \lambda_-^t)} = \log \frac{(1-\sqrt{p/q})+(1+\sqrt{p/q})(\lambda_+/\lambda_-)^t}{(1-\sqrt{q/p})+(1+\sqrt{q/p})(\lambda_+/\lambda_-)^t} .
\end{align*}
and for $s>0$,
\[
\lim_{t \to \infty} \dec_t^{(2)} = \log \frac{1+\sqrt{p/q}}{1+\sqrt{q/p}} = \log \sqrt{p/q} = \frac{1}{2} \log \frac{p}{q},
\]
since $|\lambda_+|>|\lambda_-|$.
For our choice of parameters, $\sqrt {p/q} = \sqrt{e}$ so $\lim_{t \to \infty} \dec_t^{(2)} = \frac{1}{2}$.
Note, $p/q$ measures the noisiness of the independent observations of the environment and is a proxy for the signal-to-noise ratio (SNR).  If $p \gg q$, then $\lim_{t \to \infty} \dec_t^{(2)}$ is large.  On the other hand, if observations are noisy, \emph{i.e.} $p \approx q$, then $\lim_{t \to \infty} \dec_t^{(2)} \approx 0$.

 {When all  observations are informative ($s = 0$)}, the social information will alternate indefinitely between two values. In this case, {$\lambda_+/\lambda_- = -1$}, and we have
\[
\dec_t^{(2)} = \log \frac{(1-\sqrt{p/q})+(1+\sqrt{p/q})(-1)^t}{(1-\sqrt{q/p})+(1+\sqrt{q/p})(-1)^t} = \left\{ \begin{array}{cc} 0, & t \ \text{even}, \\
\log (p/q), & t \ \text{odd},
\end{array} \right.
\]
since the belief of  agent 1  must be at lattice site 0 (1) after an even (odd) number of observations.


\begin{figure}[h!]
\centering
\includegraphics[width = 12cm]{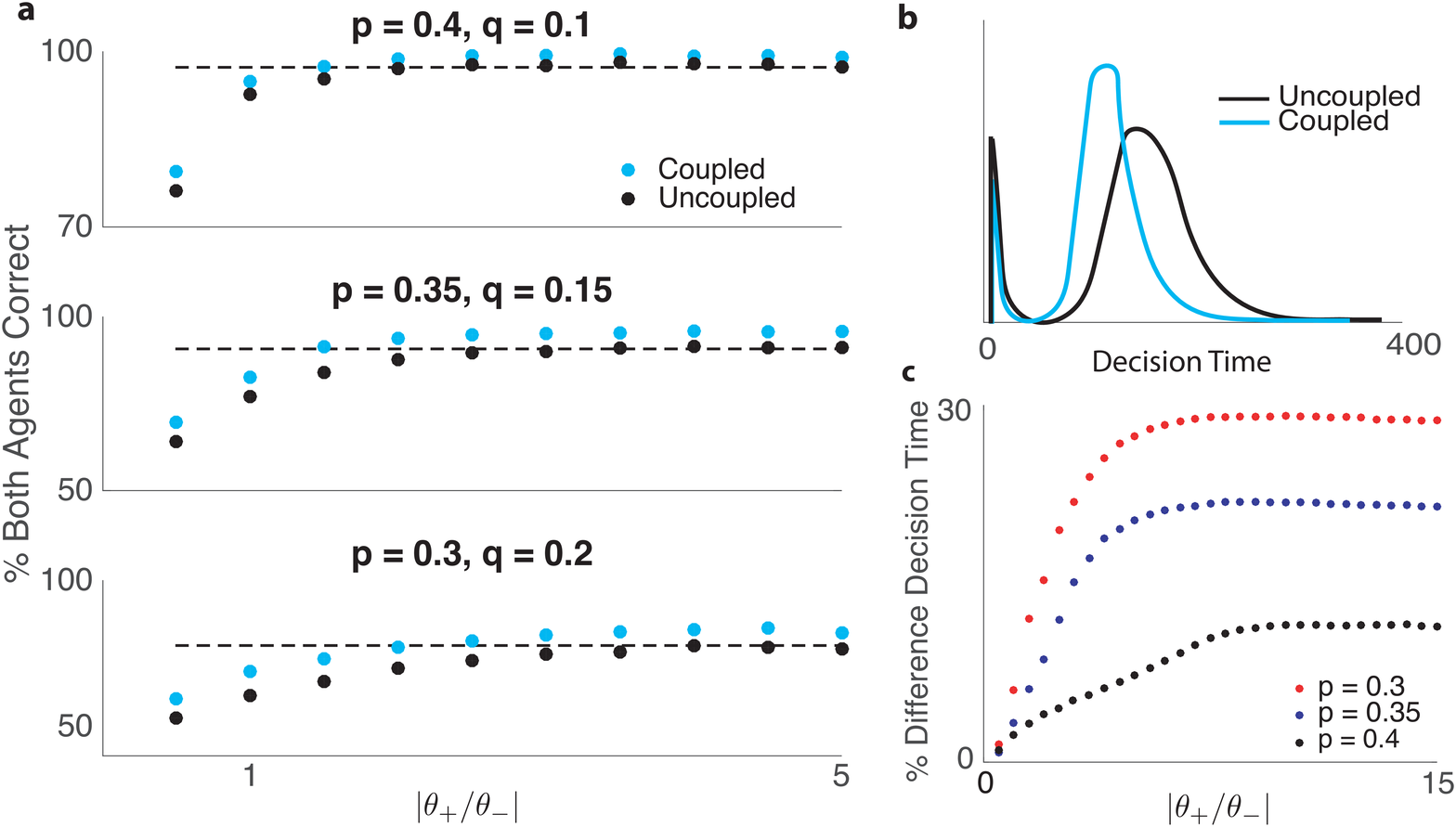}
\caption{Statistics of decisions in a two-agent network with (coupled), and without (uncoupled) communication from agent 1 to 2.  The lower threshold is fixed at $\theta_- = -2 $ throughout. (a) The fraction of times both agents selected the correct choice as a function of asymmetry in the system as measured by $\theta_+ / \theta_-$.    The dashed lines are asymptotic probabilities of the correct choice in the limit $\theta_{+} \to \infty$ for uncoupled agents. (b) First passage time distributions for the LLR in the case $\theta_+ = 40$.  (c) Relative percent difference in decision times for agent 2 in coupled versus uncoupled trajectories as a function of asymmetry in the system shown for different values of $p$.  }
\label{sumstat}
\end{figure}

These formulae reveal that social evidence is typically not accrued significantly beyond the first few observations when decision thresholds are small. This is because social information is bounded above by $\frac{1}{2} \log \frac{p}{q}$, which will only be large in cases that $p/q$ is very large and decisions will almost always be rapid (occurring after two or three timesteps). We show this for $\thm = - 1$ and $\thp = 2$ in Fig.~\ref{fig:fourStateExample}a, with $\soc_t^{(1)} \to 1/2$ as $t \to \infty$ for our choice of parameter values.

In general, social information, $\soc_t^{(1)},$ will converge to a value on the side of the threshold with larger absolute value in the absence of a decision. Intuitively,  when $|\thp| > |\thm|$ then more observations and a longer time are required for an agent's belief to reach $\thp$, on average.  In this case, observing that agent 1 has not chosen $\hm$ after a small initial period suggests that this agent has evidence favoring $\hp$.

To illustrate the impact of asymmetry in the measurement distributions, we varied the probability, $p,$ of an observation favoring $H^+$, while keeping the increments in the belief fixed.  When agent 2 observes the decisions of agent 1, the probability that both reach the correct decision is larger than when they both gather information independently (See Fig.~\ref{sumstat}a).  In particular, as $p/q$ decreases so that private observations provide less information,  the social information has an increasing impact on accuracy. 

Social information also affects decision times, particularly in the case of strongly asymmetric thresholds (Fig.~\ref{sumstat}b). An early peak in the distributions represents decisions corresponding to the smaller threshold, $\thm$, while the latter peak corresponds to the opposite decision when the belief crosses $\theta_+ \gg - \thm$. 
As $p/q$ increases, the difference in decision times between the agents decreases (See Fig.~\ref{sumstat}c), as social information has a larger impact on decisions when private measurements are unreliable. 

\noindent
{\bf Remark:} The impact of social information in this example is small, unless the difference in thresholds is very large.
However, this impact can be magnified in larger networks: Consider a star network in which an agent observes the decision of $N > 1$ other agents.  If these are the only social interactions, the independence of private measurements implies that social information obtained by the central agent is additive. Until a decision is observed, this social information  equals $N \soc_t^{(1)},$ where $\soc_t^{(1)}$ is as defined in Eq.~\eqref{social}.  Thus the impact of non-decision information can be magnified in larger networks. However, as these cases are computationally more difficult to deal with, we do not discuss them in detail here. 

\section{Two-agent recurrent network}
\label{sec:bidsymm}
We next consider two agents that can observe and react to each other's choices. We again assume that at each time the agents synchronously make independent, private observations, $\xi_t^{(i)} \in \Xi,$ and update their beliefs. The agents then observe  each other's decision state, $d_t^{(j)}$ ($j \neq i$) and use this social information to update their belief again. Knowing that a belief has been updated and observing the resulting decision state can provide new information about an agent's belief.  {Importantly, we assume agents add the information obtained from private observations first and then check whether or not they have sufficient information to make a decision. If not, agents gather social evidence from other agents in the network.}

{
Unlike previous cases we considered, social information exchange can occur over several steps.
We assume that agents exchange social information until they cannot learn anything new from each other 
and then make their next private observation (See Fig.~\ref{recurrent}). Alternatively, we could assume
that each social information exchange is followed by a private observation, but the analysis would proceed similarly.  
}

\begin{figure}[t!]
\centering
\includegraphics[width = 12cm]{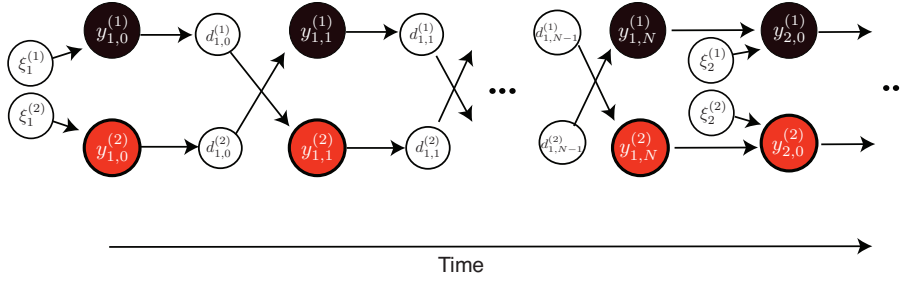}
\caption{In a recurrent network of two agents the LLRs, $y_{t,n}^{(1)}$ and $y_{t,n}^{(2)}$, of the two observers are updated
recursively. Agents update their belief after private observations, $\xi^{(i)}_t$, as well as observations of the subsequent series of  decision 
states, $d_{t ,n}^{(j)}$, of their neighbor ($j \neq i$).  After a finite number of steps, $N_t$, no further information can be obtained by observing each other's decision, and the two agents make their next private observation, $\xi^{(i)}_{t+1}$, synchronously. The process continues until one of the agents makes a decision.}
\label{recurrent}
\vspace{-3mm}
\end{figure}

We describe several properties of this evidence exchange process: As in the case of a unidirectional network, social information is additive. It evolves deterministically up to the time of a decision of the observed agent.  Once a decision is made, the social information that is communicated approximately equals the belief threshold ($\thp$ or $\thm$) crossed by the LLR of the deciding agent.   We also show that the exchange of social information after an observation terminates in a finite number of steps either when indecisions provide no further social information or when one of the agents makes a choice. 

Such information exchange has been discussed in the literature on \emph{common knowledge} and rational learning in social networks~\cite{milgrom1982}. This body of work shows that rational agents that repeatedly announce their preferred choice must eventually reach agreement~\cite{aumann1976,Geanakoplos1982,geanakoplos1994,mueller2013general,fagin2004}. Typically, it is assumed that information is shared by announcing a preference that can be changed as more information is received. Our assumption  that decisions are immutable means that agreement is not guaranteed.

We describe the process of social exchange inductively, describing the basic case following an observation at $t = 1$ in some detail. Following exchanges are similar, as the belief is updated recursively.

{\bf Social information exchange after the first observation.} After the first private observation, $\xi_1^{(i)} \in \Xi,$
at $t = 1$, the beliefs of the agents are $y_{1, 0}^{(1)}, y_{1, 0}^{(2)},$ where  $y_{1, 0}^{(i)} = \obb i1$.  Let 
\begin{equation} \label{E:Pi}
\Pi = \left\{ \LLR (  \xi ) \; | \;  \xi \in  \Xi \right\}
\end{equation}
be the set of all possible increments due to a private observation,  so that $\obb ij \in \Pi, i = 1,2.$  As the set of possible observations,  $\Xi,$ is finite,  so is $\Pi$.  We will describe how the two agents exchange social information with their neighbor until one observes a decision or no further information can be exchanged. 
The second subscript, $n$ in $y_{1, n}^{(j)},$ denotes the steps in this subprocess of social information exchange preceding a decision or subsequent private measurement.

We again associate a belief, $y_{1, n}^{(i)},$ with a corresponding decision state,
$d_{1 ,n}^{(i)}$, as in Eq.~\eqref{E:decomposition}.
If neither of the first two private observations leads to a decision, then $d_{1 ,0}^{(i)} =  0,$ and $y_{1, 0}^{(i)} \in \Theta $ for $i = 1,2$, where $\Theta \equiv (\thm, \thp)$ . Importantly, the fact that agent $i$ observed that its counterpart did not decide means that they know $y_{1, 0}^{(j)} \in \Theta,$ for $i \neq j.$  

To update their belief, agents compare the probability of all available evidence under the two hypotheses,  $P(  \xi^{(i)}_1, d_{1,0}^{(j)}  | \hp)$, and $P(  \xi^{(i)}_1, d_{1,0}^{(j)}  | \hm).$  As $d_{1,0}^{(j)}$ is independent of $\xi_1^{(i)}$ for $i \neq j$, their updated beliefs are
\begin{align*}
y_{1,1}^{(i)} &= \LLR (  \xi^{(i)}_1, d_{1,0}^{(j)}) = 
\LLR (  \xi^{(i)}_1 ) +\LLR (d_{1,0}^{(j)} = 0) \\
&=  y_{1, 0}^{(i)}  + \LLR (d_{1,0}^{(j)} = 0)
=   \obb i1 + \underbrace{\LLR (y_{1,0}^{(j)}) }_{\soc_{1,1}} =  \obb i1 +  \soc_{1,1}.
\end{align*}
We omit the superscripts on the social information, since  $\soc_{1,1}^{(1)} = \soc_{1,1}^{(2)}$ as the agents are identical. Since the agents know the measurement distributions, $f_{\pm}(\xi)$, the survival probabilities,  $P(y_{1,0}^{(i)} \in \Theta | \hp),$ can be computed as in Section~\ref{sec:unidir}.  

If $y_{1, 1}^{(1)}, y_{1, 1}^{(2)} \in \Theta,$ no decision is made after the first exchange of social information, and $d_{1,1}^{(i)} = 0$ for $i = 1,2.$ In this case, agent $i$ knows that $ y_{1,1}^{(j)} \in \Theta,$ for $j \neq i$, so $\thm  < y_{1,0}^{(j)} + \soc_{1,1}  < \thp$.  {As the initial private observation of agent $j$ did not lead to a decision, agent $i$ also knows $y_{1,0}^{(j)} \in \Theta$.} Thus observing $d_{1,1}^{(j)} =0$ informs agent $i$  that $y_{1,0}^{(j)} \in \Theta 
\cap (\Theta - \soc_{1,1}) \equiv \Theta_{1,1},$ for $i \neq j$.  More precisely, 
\begin{equation} \label{eq:step1}
P( d_{1,1}^{(j)} = 0 |  d_{1,0}^{(i)} = 0 , \hp) = P( y_{1,0}^{(j)} \in \Theta_{1,1}| \hp ).
\end{equation}

Any initial measurement $ \xi^{(j)}_1$ that led to a belief $y_{1,0}^{(j)} \notin \Theta_{1,1}$  would have {also led} to a decision at this point.  This would end further evidence accumulation.  Thus the other agent either observes a decision or knows that  $y_{1,0}^{(j)} \in \Theta_{1,1}$.  

We will deal with the consequence of observing a decision subsequently.
If an observation does not lead to a decision after the first exchange of social information, then  $y_{1,0}^{(j)} = \obb i1 \in \Theta_{1,1}$
implies that
\begin{align*}
y_{1,2}^{(i)} &= 
\LLR (  \xi^{(i)}_1, d_{1,0}^{(i)} = 0, d_{1,1}^{(i)} = 0, d_{1,0}^{(j)} = 0, d_{1,1}^{(j)} = 0 ) =  \LLR(\xi^{(i)}_1) + \LLR (d_{1,1}^{(j)} = 0 |  d_{1,0}^{(i)} = 0) \\
&= \obb i1 + \LLR ( y_{1,0}^{(j)} \in \Theta_{1,1}) := \obb i1 + \soc_{1,2}.
\end{align*}
Again, if $y_{1,2}^{(i)} \in \Theta$, neither agent makes a decision, and both will observe $d_{1,2}^{(i)} = 0$.

To extend this argument we define $ \Theta_{1,l-1} \equiv \bigcap_{n = 0}^{l-1} (\Theta - \soc_{1,n})$ 
so that $\obb i1 \in \Theta_{1,l-1}$ implies that agent $i$ has not made a decision at step $l-1$ of the
social exchange, $d_{1,l-1}^{(i)} = 0$.  In this case we have,
\begin{align*}
&P(  \xi^{(i)}_1, d_{1,0}^{(i)} = 0, \ldots d_{1,l}^{(i)} = 0, d_{1,0}^{(j)} = 0,\ldots d_{1,l}^{(j)} = 0   | \hp  ) \\
&= P(  \xi^{(i)}_1,  d_{1,l-1}^{(i)} = 0,  d_{1,l}^{(j)} = 0   | \hpm  ) =  P(  \xi^{(i)}_1  | \hp )  P( d_{1,l}^{(j)} = 0 |  d_{1,l-1}^{(i)} = 0 , \hp).
\end{align*} 
Thus if the $l$-th exchange of social information results in no decision, each agent 
updates its belief  as 
\begin{equation} \label{E:social}
 y_{1,l}^{(i)}  =  \obb i1+  \LLR ( y_{1,0}^{(j)} \in \Theta_{1,l-1} ) \equiv \obb i1 + \soc_{1,l},
\end{equation}
where $\soc_{1,0} = 0$.

This exchange of social information continues until one of the agents makes a choice or
the absence of a decision does not lead to new social information~\cite{Geanakoplos1982}. The second case occurs at a step, $N_1$, at which indecision provides no further information about the belief of either agent, so that
\begin{equation} \label{E:terminate}
\Pi \cap \Theta_{1,N_1} = \Pi \cap \Theta_{1,N_1+1},
\end{equation}
where $\Pi$ is  defined in Eq.~\eqref{E:Pi}.
In this case, $\soc_{1,N_1} = \soc_{1,N_1+1},$ and, if {neither agent decides}, their beliefs would not change 
at step $N_1+1$.  As both agents know that nothing new is to be learned from observing their neighbor, they then make the next private observation.

We denote the total social information gained after this exchange is complete by $\soc_{1} = \soc_{1,N_1},$ and the resulting belief by $y_{1}^{(i)} =  y_{1,N_1}^{(i)}= \obb i1 + \soc_{1}$. If no decision is made at this point, then agent $j$ 
knows that $y_{1,0}^{(i)} = \obb i1 \in \left(\bigcap_{n = 0}^{N_1} (\Theta - \soc_{1,n})\right) \equiv \Theta_1$.

The process can be explained simply: The absence of a decision provides sequentially tighter bounds on the neighbor's belief. When the agents can conclude that these bounds do not change from one step to the next, the absence of a decision provides no new information, the exchange of social information ends, and both agents make the next private observation.

Importantly, this process is \emph{deterministic}:  until a decision is made, the social information gathered on each step is the same across realizations,
\emph{i.e.} independent of the private observations of the agent.


{\bf Social information exchange after an observation at time $t>1$.} The integration of private information from each individual measurement, $\xi_t^{(i)}$, is again followed by an equilibration process. The two agents observe each others' decision states until nothing further can be learned.
To describe this process, we proceed inductively, and extend the definitions introduced in the previous section.

 Let $\obb it  = \LLR (  \xi^{(i)}_t )$ be the private information obtained from an observation at time $t$.  {For the inductive step we assume that each equilibration process after the observation at time $t$ ends either in a decision or allows each agent $i$ to conclude that the accumulated social and
private evidence  of agent $j$ were insufficient to cross a decision threshold.  We will see that this means that ${\rm Priv}_{1:t}^{(j)} \in \Theta_{1:t}$ where  $\Theta_1$ was defined above, and we define the other sets in the sequence recursively.}

Following equilibration after observation $\xi^{(i)}_1,$ either one of the agents makes a decision, or  
each agent $i$ knows that  the private information of its counterpart satisfies $\obb j1 \in \Theta_1$ for $j \neq i$.
Let
$$
\Theta_{1:t} = \left\{ {\rm Priv}_{1:t}^{(j)} \Big | \sum_{k = 1}^l  {\rm Priv}_{k}^{(j)} \in \Theta_l \text{ for all }  1 \leq l \leq t \right\}.
$$
Thus, agent $i \neq j$ knows that any sequence ${\rm Priv}_{1:t}^{(j)}$ that did  not lead to a decision by agent $j$ must lie in $\Theta_{1:t}$. 
To define the social information gathered during equilibration following the observation at time $t$, let $\Theta_{t,0} = \Theta \cap (\Theta - $Soc$_{t-1})$,   $\soc_{t,0} = \soc_{t-1},$  and set
\begin{equation} \label{E:social_gen}
\soc_{t,l} :=  \LLR \left( {\rm Priv}_{1:t-1}^{(i)} \in \Theta_{1:t-1}, \text{ and }  \sum_{k = 1}^t  {\rm Priv}_{k}^{(j)} \in \Theta_{t,l-1} \right), 
\end{equation}
for $l \geq 1$, where $ \Theta_{t,l-1} \equiv \bigcap_{n = 0}^{l-1} (\Theta - \soc_{t,n})$.

\begin{theorem}
\label{thm:integration}
Assume that, in a recurrent network, two agents have made a sequence of private observations, $\xi_{1:t}$, followed by $l$ observations of the subsequent decision states of each other. If neither agent has made a decision, then the belief of each  is given by 
\begin{equation} \label{E:general}
y_{t,l}^{(i)} = \sum_{k = 1}^t \obb ik + \soc_{t,l},
\end{equation}
for $1 \leq t, i = 1,2$. The exchange of social information terminates in a finite number
of steps after an observation, either when a decision is made or after no further social information
is available at some step $l = N_t$.  The private evidence in Eq.~\eqref{E:general} is a random variable (depends on realization), while the social evidence is
independent of realization and equal for both agents.
\end{theorem} 

\begin{proof}  
{We prove the theorem using double induction.  The outer induction is on the times of the private observations, $t$, while
the inner induction is on the steps in the equilibration process. The basic case of the induction over $t$ was  proved in the previous section.} 
Induction over $t$ in this case follows similarly.  To address the inner steps, by a slight abuse of notation, let 
$d_{1:t, 1:l}^{(i)}$ denote the sequence of decision states up to the $l$-th step in the equilibration process. If no 
decision has been made at this step, we write $d_{t, l}^{(i)} = 0$.  This implies that no previous decision
has been made, and we denote this by writing $d_{1:t, 1:l}^{(i)} = 0.$

At $l=0$, we assume that equilibration following
 private observation $\xi_t^{(i)}$ terminates on step  $N_t$.  Conditional independence of the observations implies that
$P(\xi_{1:t+1}^{(i)}, d_{1:t, 1:N_t}^{(i)} = 0, d_{1:t, 1:N_t}^{(j)} =0 | \hpm) = P(\xi_{t+1}^{(i)} | \hpm) P(\xi_{1:t}^{(i)},  d_{1:t, 1:N_t}^{(i)} = 0, d_{1:t, 1:N_t}^{(j)} =0 | \hpm), $
so that  
$$
y_{t,0}^{(i)} = \LLR (\xi_{1:t+1}^{(i)}, d_{1:t, 1:N_t}^{(i)} = 0, d_{1:t, 1:N_t}^{(j)} =0 )  =  \sum_{k = 1}^{t+1} \obb ik + \soc_{t+1,0},
$$
where we used $\soc_{t+1,0} = \soc_t.$

Suppose no decision is made in the following $l \geq 0$ equilibration steps, so that $d_{1:t+1, 1:l}^{(i)} = 0$ for 
$i = 1,2$. For all sequences of measurements $\xi_{1:t+1}^{(i)}$ that are consistent with this absence of a decision, we 
 have 
 $$
 P(\xi_{1:t+1}^{(i)}, d_{1:t+1, 1:l}^{(i)} = 0, d_{1:t+1, 1:l}^{(j)} =0 | \hpm) = P(\xi_{1:t+1}^{(i)}, d_{t+1, l-1}^{(i)} = 0, d_{t+1, l}^{(j)} =0   | \hpm),
 $$
 and therefore
 $$
 P(\xi_{t+1}^{(i)}, d_{t+1, l-1}^{(i)} = 0, d_{t+1, l}^{(j)} =0   | \hpm) = 
P(  \xi_{1:t+1}^{(i)}  | \hpm )  P( d_{t+1,l}^{(j)} = 0 |  d_{t+1,l-1}^{(i)} = 0 , \hpm).
$$
It follows that
$$ 
y_{t,l+1}^{(i)} = \LLR (  \xi_{1:t+1}^{(i)}) + \LLR ( d_{t+1,l}^{(j)} = 0 |  d_{t+1,l-1}^{(i)} = 0) =  \sum_{k = 1}^{t+1} \obb ik + \soc_{t+1,l+1},
$$
where $\soc_{t+1,l+1}$ is defined in Eq.~\eqref{E:social_gen}. This exchange of social information stops at $l = N_{t+1}$ 
when $ \Pi \cap \Theta_{t+1, N_{t}} = \Pi \cap \Theta_{t+1,N_{t}+1}$, and neither agent learns anything further from the absence of a decision by its counterpart.
\end{proof}

{\bf Belief update after a decision.}  The following proposition shows what happens when the belief of one of the agents crosses a threshold. 
 
\begin{proposition}
\label{prop:bump2}
Suppose that in a recurrent, two-agent network agent $i$ makes a decision after a private observation at time $T$ during the $n$th step of the subsequent social information exchange process, $d_{T,n}^{(i)} = + 1$. Then agent $j \neq i$, updates its belief as
\begin{align*}
y_{T,n+1}^{(j)} = \obb {j}{1:T} + \soc^+_{T,n+1} , \ \ \ \ \text{with} \ \ \ \ 
\theta_+ - \soc_{T,n}  <  \soc^+_{T,n+1} < \theta_+ - \soc_{T,n} + \ep_{T,n}^+,
\end{align*}
where we can bound
\begin{align*}
\ep_{T,0}^+ \leq  \sup_{\xi_{1:T} \in {\mc C}_+(T,0)}\LLR (\xi_{T}), \ \ \ \ \text{and} \ \ \ \ \ep_{T,n}^+  \leq (\soc_{T,n}-\soc_{T,n-1}), \ \ n>0,
\end{align*}
 and ${\mc C}_+(T,0)$ is the set of all chains of observations leading an agent to choose $H^+$ at timestep $(T,0)$.  An analogous result holds for  $d_{T,n}^{(i)} = - 1$. 
\end{proposition}

%

\begin{figure}[t!]
\hspace{1cm}\includegraphics[width = 12cm]{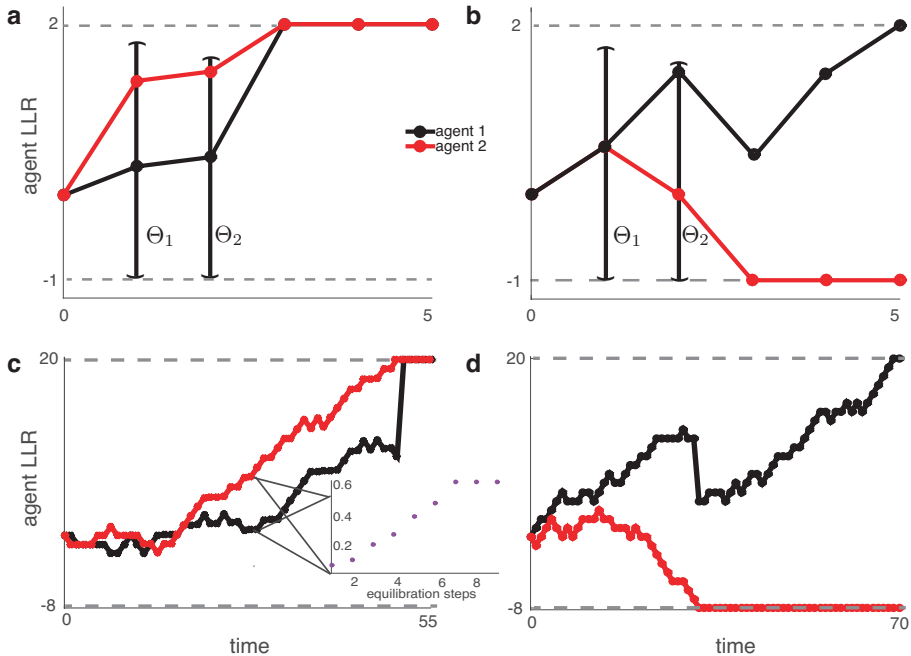}
\caption{The belief (LLR) of  agents 1 and 2 in a recurrent network with asymmetric thresholds when (a) the agents agree,  and (b) agents disagree on a final choice.  Also shown are the intervals $\Theta_1, \Theta_2$,  resulting from the equilibration following the first two observations. Each agent knows that the LLR of its counterpart lies in this interval after equilibration ends.  Although the beliefs evolve stochastically, the sequence $\Theta_1, \Theta_2$ is fixed across realizations.  Here, $p = e/7$, $q = 1/7$, and $s$ is determined from $s = 1 - p - q$.  Also, $\theta_+ = 2$ for (a)-(b) and $\theta_+ = 20$ for (c) - (d).  $\theta_-  = -1$ for (a)-(b) and $\theta_-= -8$ for (c)-(d).  Inset:  Social information obtained from equilibration converges in seven steps at the indicated time. {Here, $H = H^+$.} }
\label{equil1}
\end{figure}

A proof of Proposition~\ref{prop:bump2} is given in Appendix~\ref{bnd2coup}. This proposition shows that any social information obtained before observing a decision is subsumed in the information obtained from observing a decision {and that the social information acquired from a neighbor's decision is roughly that neighbor's private evidence. After one of the agents makes a decision, the other agent continues gathering private observations until they make their own decision.}

If the two agents have different decision thresholds, the expression for the post-decision belief is more involved but still computable. For simplicity we forgo further discussion of this case. On the other hand, when the  the thresholds and evidence distributions are symmetric, the evidence accumulation process is much simpler.

\begin{proposition}
When the distributions $f_+$ and $f_-$ are symmetric and the agents have the same symmetric thresholds ($\pm \theta_{\pm} = \theta$), then $ \soc_{t} = 0$ until private evidence leads one of the agents to make a decision. Thus there is no exchange of social information until one of the agents makes a decision.
\end{proposition} 

\begin{proof}
The argument is similar to that used to prove Proposition~\ref{prop:symmetric}. We can proceed inductively again: If the two agents have not made a decision after the first observation, by symmetry this does not provide any evidence for either hypothesis $H = H^{\pm}$. Observing each other's decisions after this first observation hence results in no new information,
\[
\soc_{1,1}= \LLR (y_{1,0}^{(j)} \in \Theta ) = 0.
\]
Therefore, the update process terminates immediately, and both agents proceed to the next observation.  
As shown  in Proposition~\ref{prop:symmetric}, further observations provide no new social information unless a decision is made.  The two agents therefore continue to accumulate private information until one of them makes a decision. 
\end{proof}

Fig.~\ref{equil1} provides examples of evidence accumulation in a two-agent recurrent network.  In Fig.~\ref{equil1}a,b, we illustrate the process with relatively small thresholds and show how the intervals $\Theta_n$ shrink at the end of each social evidence exchange (equilibration) following a private observation. Note that the sequence of intervals is the same in both examples because the social information exchange process is deterministic.  In this example equilibration ends after two steps.  Figs.~\ref{equil1}c,d,  provide an example with strongly asymmetric decision thresholds.   These examples also show that  the beliefs of the two agents do not have to converge, and  the agents do not need to agree on a choice, in contrast to classic studies of {\em common knowledge}~\cite{aumann1976,geanakoplos1994,milgrom1982}. 


\section{Accumulation of Evidence on General Networks}
 In networks that are not fully connected, rational agents need to take into account the impact of the decisions of agents that they do not observe directly. To do so they marginalize over all unobserved decision states.  This computation can be complex, even when thresholds and evidence distributions are symmetric. 
 
 To illustrate we begin by describing the example of  agents with symmetric decision thresholds and measurement distributions on a directed chain. Symmetry makes the computations more transparent, as  the absence of a decision is not informative about the belief of any agent, hidden or visible. Social information is therefore only communicated when an agent directly observes a choice. Such an observation leads to a jump in the agent's belief and can initiate a cascade of decisions down the chain~\cite{caginalp2017decision}. 

We note that once an agent in the network makes a decision, symmetry can be broken:  agents know that all others who have observed the decision have evidence favoring the observed choice.  {This results in a network evidence accumulation process akin to one when agents have asymmetric thresholds.}   As we have seen {in Sections \ref{sec:unidir} and \ref{sec:bidsymm}}, once symmetry is broken, even the absence of a choice provides social information, leading to belief equilibration. We therefore only describe evidence accumulation up to a first decision.

\subsection{Terminology and Notation}

In a social network of $N$ agents, we again assume that each agent makes a private observation at every time step. After incorporating the evidence from this observation the agent then updates its decision state and shares it with its neighbors. A directed edge from agent $j$ to $i$, denoted by $j \to i$, means  that agent $j$ communicates its decision state to agent $i$. The set of neighbors that agent $i$ observes is denoted by ${\mc N}^{(i)}\ $.

Agent $i$ receives social information from all agents in ${\mc N}^{(i)}$, but must  also  take into account decisions of \emph{unobserved} agents. We define the set of all agents not visible to agent $i$ as
\[
{\mc U}^{(i)} = \{ j :  j \notin {\mc N}^{(i)}, \text{ and } j \neq i  \} .
\]
Thus all agents in the network, besides agent $i$, belong to ${\mc N}^{(i)}$ or ${\mc U}^{(i)}$.  {Therefore $\mathcal{N}^{(i)}, \mathcal{U}^{(i)},$ and $\{i\}$ form a partition of all nodes in the network.}

We denote the set of decisions by the neighbors of agent $i$ following the observation at time $t$ by $d^{{\mc N}^{(i)}}_t = \{ d_{t}^{(k)} : k \in {\mc N}^{(i)} \}$.  Similarly, the set of the decisions by unobserved agents is denoted by $d^{{\mc U}^{(i)}}_t$. More generally,  $d^{{\mc N}^{(i)}}_{1:t}$ denotes the sequence of decision states of the neighbors of agent $i$ up to and including the decision following the observation at time $t$: $d^{{\mc N}^{(i)}}_{1:t} = \{ d^{{\mc N}^{(i)}}_s : 1 \leq s \leq t\}$.   We will see that in the case of symmetric thresholds and observations no equilibration occurs, so these decision states describe information available to each agent in the network completely until one of the agents makes a choice. At time $t$, the private and social observations obtained by agent $i$ are therefore 
$I_t^{(i)} = \{ \xi^{(i)}_{1:t} , d^{{\mc N}^{(i)}}_{1:t-1} \}$.    As before, we denote the private information by $\obb i{1:t} = \LLR ( \xi^{(i)}_{1:t})$.

\subsection{Non-decisions}
\label{sec:predec}
{In a social network of agents with symmetric thresholds and measurement distributions, two  properties of decision information simplify computations: an absence of a decision is uninformative, and 
once a decision is observed, the resulting social information is additive. Therefore,
$$
y_t^{(i)} = \LLR ( \xi^{(i)}_{1:t}, d^{\mathcal{N}^{(i)}}_{1:t}=0) = \text{Priv}^{(i)}_{1:t}.
$$
The proof of this equality follows an argument similar to that given in Sections~\ref{sec:unidir} and~\ref{sec:bidsymm}.  The main difference is that in the present case each agent marginalizes over unobserved decision states.  Symmetry implies that for each decision history of unobserved agents, $d_{1:t}^{\mathcal{U}^{(i)}} $, there is a corresponding  opposite decision history, $-d_{1:t}^{\mathcal{U}^{(i)}} $, obtained by flipping the sign of each decision in the vector $d_{1:t}^{\mathcal{U}^{(i)}} $ and leaving non-decisions unaffected.   By symmetry, both decision histories are equally probable, and their contributions cancel when agent $i$ computes its belief. Hence no agent $i$  obtains information from observing the absence of a decision by any of its neighbors, and equilibration  never starts.  As a result agents independently accumulate evidence, and their neighbors' decision states only become informative when one of the neighboring agents makes a choice. As in the case of two agents, a choice by any agent will lead to a jump in the belief of all observing neighbors. }\\

\subsection{Example of Marginalization:  3-Agents on a Line}
\label{sec:line3}
We demonstrate the computations needed to account for unobserved agents using the example shown in Fig.~\ref{fig:3line}. In this case a decision by agent $3$ is not observed by any other agent. {Agent 1 and 2 therefore accumulate information as described in Section~\ref{sec:uni}.  To understand decisions in this network it is therefore sufficient to only consider the computations agent 3 must perform upon observing a decision by agent 2.  We will see that observing a decision in a network that is not fully connected can 
result in a jump in the observer's LLR that is greater than the the threshold of the observed decision.  The additional information  comes from the possibility that decisions of unobserved agents could have lead to the decision of an observed agent.}

\begin{figure}[t!]
\centering
\includegraphics[width = 8cm]{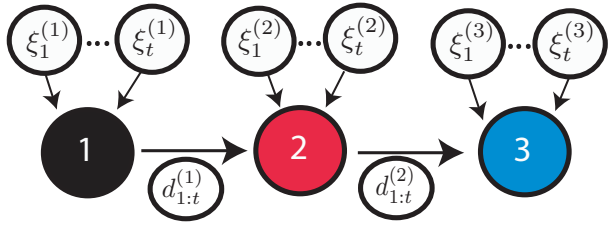}
\caption{Agents on a line.  Agent 3 observes the decisions of agent 2, but not of agent 1. However, agent 3 still needs to take into account the possibility that a choice by agent 1 has caused a decision by agent 2.}
\label{fig:3line}
\end{figure}

Agent 3 has no social information before observing a decision by agent 2. After observing a decision  by agent 2 at time $T^{(2)}$, agent 3 updates its belief by marginalizing over possible decision histories of agent 1: 
\begin{equation} \label{E:decomposition2}
\begin{split}
P(d_{T^{(2)}}^{(2)} = \pm 1 | H^{\pm}) &=  P(d_{T^{(2)}}^{(2)} = \pm 1 , d_{T^{(2)}}^{(1)} = 0| H^{\pm})  +  \sum_{s =1}^{T^{(2)}} \sum_{d = \pm 1} P(d_{T^{(2)}}^{(2)} = \pm 1 , d_{s-1}^{(1)} = 0, d_s^{(1)} = d| H^{\pm}).
\end{split}
\end{equation}
A choice of agent 2 can be triggered by either: (a) a private observation leading to the belief $y_t^{(2)}$ reaching one of the thresholds, $\theta_{\pm}$, or (b) the decision of agent 1 causing a jump in the belief 
$y_t^{(2)}$ above threshold $\theta_+$ or below threshold $\theta_-$. Possibility (b) is captured by the term 
\begin{align}
P(d_{T^{(2)}}^{(2)} = \pm 1 , d_{T^{(2)}-1}^{(1)} = 0, d_{T^{(2)}}^{(1)} = \pm 1| H^{\pm}) \label{1trig2}
\end{align}
in Eq.~\eqref{E:decomposition2}, while possibility (a) corresponds to all the other summands.

An argument equivalent to that in Proposition~\ref{prop:bump} shows that the social information communicated in case (a) is close to $\pm \theta$ for $d_{T^{(2)}}^{(2)} = \pm 1$.  However, in the second case the belief of agent 2 satisfies $y_{T^{(2)}}^{(2)} \in [\theta, 2 \theta]$ for $d_{T^{(2)}}^{(2)} = + 1$ or $y_{T^{(2)}}^{(2)} \in [-2\theta, -\theta]$ for $d_{T^{(2)}}^{(2)} = - 1$, modulo a small correction.  Agent 3 updates its belief by weighting both possibility (a) and (b), and hence increases its belief by an amount greater than $\theta$. 

To show this, we note that for any $s=1,..., {T^{(2)}}-1$,
\begin{align}
\frac{P(d^{(2)}_{T^{(2)}} =  1, d^{(1)}_{s} = \pm1|\hp)}{P(d^{(2)}_{T^{(2)}} = 1, d^{(1)}_{s} = \pm1|\hm)} \approx \frac{P(d^{(2)}_{T^{(2)}} =  1, d^{(1)}_{s} = 0|\hp)}{P(d^{(2)}_{T^{(2)}} = 1, d^{(1)}_{s} = 0 |\hm)} \approx e^{\theta},
\label{nicecase}
\end{align}
where the last approximation follows from again assuming  ${\rm Priv}_t^{(2)}$ is small. The situations where a decision from agent 1 does {\em{not}} immediately cause a decision from agent 2, conditioned on $H^+$, are described by  
\begin{align*}
\zeta_{T^{(2)}}^+ := & P(d_{T^{(2)}}^{(2)} =  1 , d_{T^{(2)}}^{(1)} = 0| \hp) +  \sum_{s = 1}^{T^{(2)}} P(d_{T^{(2)}}^{(2)} =  1 , d_s^{(1)} = -1, d_{s-1}^{(1)} = 0| \hp)  \\
&  + \sum_{s=1}^{{T^{(2)}}-1} P(d_{T^{(2)}}^{(2)} = 1 ,  d_s^{(1)} = 1,  d_{s-1}^{(1)} = 0| \hp).
\end{align*}
Using Eq.~(\ref{nicecase}), we have
\begin{align*}
\zeta_{T^{(2)}}^+ \approx & \ e^{\theta} \cdot \left[ P(d_{T^{(2)}}^{(2)} =  1 , d_{{T^{(2)}}-1}^{(1)} = 0| \hm) +  \sum_{s = 1}^{T^{(2)}} P(d_{T^{(2)}}^{(2)} =  1 , d_s^{(1)} = -1, d_{s-1}^{(1)} = 0| \hm) \right.  \\
&  \left. +  \sum_{s = 1}^{T^{(2)}-1} P(d_{T^{(2)}}^{(2)} = 1 ,  d_s^{(1)} = 1, d_{s-1}^{(1)} = 0| \hm) \right].
\end{align*}
By Eq.~\eqref{E:decomposition2}, we have
\begin{align*}
P(d_{T^{(2)}}^{(2)} =  1 | H^+) =&  \ \zeta_{T^{(2)}}^+ + P(d_{T^{(2)}}^{(2)} = 1, d_{{T^{(2)}}}^{(1)} = 1, d_{{T^{(2)}}-1}^{(1)} = 0 |\hp) \\
P(d_{T^{(2)}}^{(2)} = 1 | H^-) \approx &  \ e^{- \theta} \zeta_{T^{(2)}}^+ + P(d_{T^{(2)}}^{(2)} = 1, d_{{T^{(2)}}}^{(1)} = 1, d_{{T^{(2)}}-1}^{(1)} = 0 |\hm),
\end{align*}
so that
\begin{equation}
\soc_{T^{(2)}}^{(3)} \approx \log\Big(\frac{P(d_{T^{(2)}}^{(2)} = 1, d_{{T^{(2)}}}^{(1)} = 1, d_{{T^{(2)}}-1}^{(1)} = 0 |\hp)+\zeta}{P(d_{T^{(2)}}^{(2)} = 1, d_{{T^{(2)}}}^{(1)} = 1, d_{{T^{(2)}}-1}^{(1)} = 0 |\hm)+\zeta e^{-\theta}}\Big).
\label{soc_33}
\end{equation}
Note that
\begin{equation}
P(d_{T^{(2)}}^{(2)} = 1, {d_{T^{(2)}-1}^{(2)} = 0,\text{}} d_{{T^{(2)}}}^{(1)} = 1, d_{{T^{(2)}}-1}^{(1)} = 0 |H^{\pm}) = P(y_{T^{(2)}}^{(2)} \in [0,\theta)|H^{\pm})P(d_{{T^{(2)}}}^{(1)} = 1, d_{{T^{(2)}}-1}^{(1)} = 0 |H^{\pm}),
\label{reviewer1}
\end{equation}
and
\begin{equation}
P(d_{{T^{(2)}}}^{(1)} = 1, d_{{T^{(2)}}-1}^{(1)} = 0|\hp) \approx e^{\theta} P(d_{{T^{(2)}}}^{(1)} = 1, d_{{T^{(2)}}-1}^{(1)} = 0|\hm) \approx \frac{1}{1 + e^{-\theta}}.
\label{reviewer2}
\end{equation}
{Eq.~(\ref{reviewer1}) omits the possibility that the beliefs of agents 1 and 2 cross their respective thresholds at the same time.  However, this event has small probability when private observations are not very informative, and thresholds are reasonably large, so we do not consider this case in our calculations.}

{The first approximation in Eq.~(\ref{reviewer2}) follows from the fact that agent 1's decision $\hp$ triggers a jump in agent 2's belief approximately equal to $\theta$.  The second approximation in Eq.~(\ref{reviewer2}) is a standard result from decision-making literature and can be obtained by calculating the exit probability of a drift-diffusion process on a bounded interval through the boundary in the direction given by the drift \cite{Bogacz2006}.}  Let
\begin{equation}
R_+(t) : =  \LLR (y_t^{(2)} \in [0,\theta)|y_{t}^{(2)} \in (-\theta, \theta)) =  \LLR (y_t^{(2)} \in [0,\theta)) \geq 0.
\label{posprob}
\end{equation}
Note that we can increase $R_+(t)$ by increasing $\theta$ without changing the measurement distributions:  {larger thresholds mean that more time is required to make a decision. In turn, longer decision times mean that the belief of an agent is more likely to lie in the half interval $[0,\theta)$ when $H = \hp$, resulting in a larger $\LLR (y_t^{(2)} \in [0,\theta))$. }

{Applying Eq.~(\ref{reviewer1}) to Eq.~(\ref{soc_33}) and suppressing the notation $d_{{T^{(2)}}-1}^{(1)} = 0 $, we obtain
\[
\soc_{T^{(2)}}^{(3)} \approx \log\Big(\frac{P(y_{T^{(2)}}^{(2)} \in [0,\theta)|\hp)P(d_{{T^{(2)}}-1}^{(1)} = 1|\hp)+\zeta_{T^{(2)}}^+}{P(y_{T^{(2)}}^{(2)} \in [0,\theta)|\hm)P(d_{{T^{(2)}}-1}^{(1)} = 1|\hm) + \zeta_{T^{(2)}}e^{-\theta}}\Big).
\]
The two approximations in Eq.~(\ref{reviewer2}), and Eq.~(\ref{posprob}) then yield
\begin{align}
\soc_{T^{(2)}}^{(3)} &\approx \log\Big(\frac{P(y_{T^{(2)}}^{(2)} \in [0,\theta)|\hp)P(d_{{T^{(2)}}-1}^{(1)} = 1|\hp)+\zeta_{T^{(2)}}^+}{P(y_{T^{(2)}}^{(2)} \in [0,\theta)|\hm)P(d_{{T^{(2)}}-1}^{(1)} = 1|\hp)e^{-\theta} + \zeta_{T^{(2)}}e^{-\theta}}\Big) \nonumber \\
 &\approx \log\Big(\frac{P(y_{T^{(2)}}^{(2)} \in [0,\theta)|\hp)+\zeta_{T^{(2)}}^+(1+e^{-\theta})}{P(y_{T^{(2)}}^{(2)} \in [0,\theta)|\hm) + \zeta_{T^{(2)}}(1+e^{-\theta})}e^{\theta}\Big) \nonumber  \\
& \approx \log\Big(\frac{P(y_{T^{(2)}}^{(2)} \in [0,\theta)|\hp)+\zeta_{T^{(2)}}^+(1+e^{-\theta})}{P(y_{T^{(2)}}^{(2)} \in [0,\theta)|\hp)e^{-R_+(T^{(2)})} + \zeta_{T^{(2)}}(1+e^{-\theta})}e^{\theta}\Big). \nonumber
\end{align}
It follows that
\begin{align}
\soc_{T^{(2)}}^{(3)} &\approx \log\Big(\frac{P(y_{T^{(2)}}^{(2)} \in [0,\theta)|\hp)+\zeta_{T^{(2)}}^+ (1+e^{-\theta})}{P(y_{T^{(2)}}^{(2)} \in [0,\theta)|\hp) e^{-R_+({T^{(2)}})} + \zeta_{T^{(2)}}^+ (1+e^{-\theta})}\Big )+ \theta \geq \theta.
\end{align}
}We therefore see that increasing $R_+({T^{(2)}})$ increases the magnitude of social information received from observing a choice by agent 2.

The impact of a decision $d_{T^{(2)}}^{(2)} = -1$ can be computed similarly, yielding
\begin{align*}
\soc_{T^{(2)}}^{(3)} \approx \log\left(\frac{P(y_{T^{(2)}}^{(2)} \in (-\theta,0] |\hp) e^{-|R_-({T^{(2)}})|} +\zeta_{T^{(2)}}^- (1+e^{-\theta})}{P(y_{T^{(2)}}^{(2)} \in (-\theta,0]|\hp) + \zeta_{T^{(2)}}^- (1+e^{-\theta})} \right)- \theta \leq \theta,
\end{align*}
where
$\zeta_{T^{(2)}}^-$ and
$R_-(t)$  are defined equivalently to $\zeta_{T^{(2)}}^+$ and
$R_+(t)$.

\begin{figure}[t!]
\centering
\includegraphics[width = 15cm]{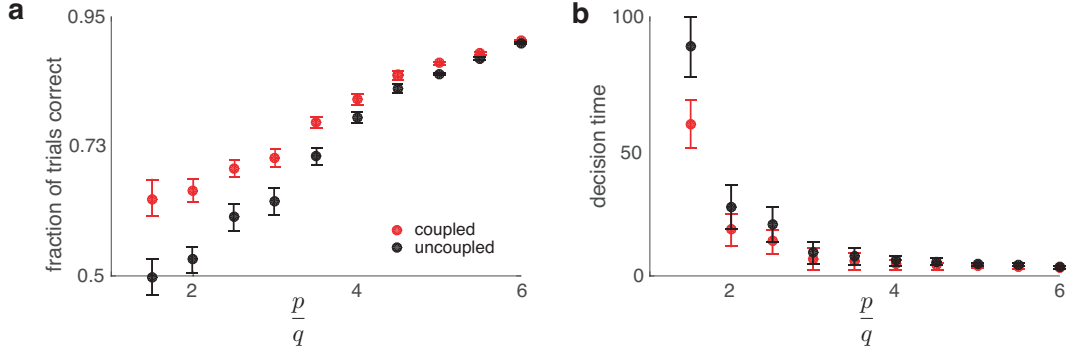}\hfill
\caption{The performance of the three agents on a directed line is better than that of independent agents.  (a) The fraction of trials for which all three agents make the correct choice.  This quantity is larger when agents are allowed to exchange information (coupled), than when agents make decisions independently (uncoupled).  (b) Time required for all three agents to make a decision.  Average decision times are also smaller when agents exchange information.    Here, $|\theta_{\pm}| = 30$, and ratio $p/q$ determines the noisiness of the measurements, as described in Section \ref{sec:4state}.}
\label{3unisumstat}
\end{figure}

In Fig.~\ref{3unisumstat}, we illustrate how  social information impacts decision-time and accuracy.  When $p/q$ is not too large, and hence the signal-to-noise ratio (SNR) of observations is low, the exchange of social information impacts both the accuracy and response time of agents significantly, and agents make the correct decision more frequently than isolated agents. When the SNR is large private information dominates decisions, and the impact of social information is small.

Marginalization becomes more complex when the unseen component, ${\mc U}^{(i)},$ of the network is larger. For instance, if we consider a long, directed line of $n$ agents, when the  $k$th agent makes a decision, the $(k+1)$st agent must marginalize over the preceding $k-1$ agents.  {The computations in this case are similar to what we presented, but more complex.  Assuming the likelihood ratio threshold for the $k$th agent choosing $\hp$ is $e^\theta$ (as in Eq.~(\ref{nicecase}) for the calculation above) and  marginalizing over the decision states of the previous $k-1$ agents  shows that the jump in the $(k+1)$st agent's belief can far exceed $\theta$.}  If the resulting jump in the belief, $y_t^{(k+1)}$, exceeds $2 \theta$, this triggers an immediate decision in agent $k+1$, and all successive agents. This is equivalent to herding behavior described in the economics literature~\cite{banerjee1992,Acemoglu2011}.

\section{Three-Agent Cliques}
\label{sec:3cliq}
In cliques, or all-to-all coupled networks, all agents can observe each others' decisions, and ${\mc U}^{(i)} = \emptyset$ for all $i$.  This simplifies the analysis, as no agent needs to marginalize over the decision states of unobserved agents.  We start by discussing the 
case of three-agent cliques in some detail, and proceed to cliques of arbitrary size in the next section.

As we are assuming symmetry, social evidence is shared only after the first agent  makes a choice.  There is a small probability that private information leads to a concurrent decision by multiple agents, and we consider this case at the end of the section.  Observing a decision by an agent may or may not drive the undecided agents to a decision.  Both the presence and absence of a decision by either remaining agent reveals further social information to its undecided counterpart.  We will see that, akin to the equilibration process, there can be a number of steps of social evidence exchange. Again the remaining agents gather all social information before the next private measurements. 

\begin{figure}[t!]
\includegraphics[width = 12.5cm]{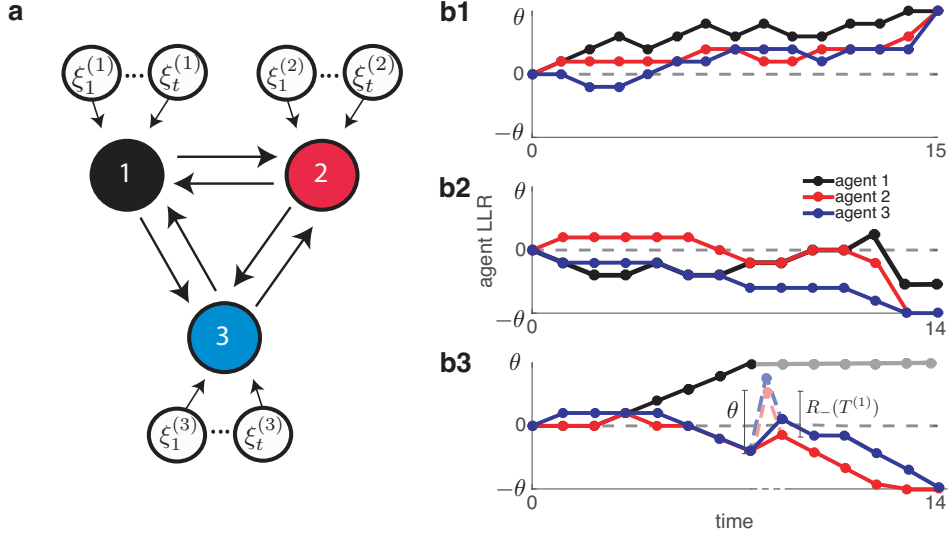}
\caption{(a) In a three-agent clique all three agents make independent observations, and observe each other's decisions.  (b) {Three main possibilities follow a decision: (b1) If the beliefs of the undecided agents are both positive, both follow the decision of the decider (agent 1 here); (b2)  If the decision causes only one of the remaining agents to decide (agent 2 here), this secondary decision leads to a further update in the belief of the remaining agent  (agent 1 here);  (b3)  If neither of the remaining agents decides, they observe each other's indecision and update their belief accordingly.  This update  cannot lead to a decision, and both agents continue to accumulate private evidence.  The dashed portion of the belief trajectory shows the intermediate steps in  social information exchange.  No decision can be reached, and the two agents continue to accumulate information, but now have further knowledge about each other's belief.  Cases in which the private evidence leads to a simultaneous decision by two or three agents are not shown.  {Here, $H = H^+$.}  }}
\label{3clique}
\end{figure} 

For concreteness we assume in the following without loss of generality that agent 1 makes a decision before agents 2 and 3.  {Agent 1's decision leads to  a jump in the beliefs of agents 2 and 3 approximately equal to the threshold corresponding to the decision.  An argument equivalent to the one presented in the previous section shows that  $d^{1}_{T^{(1)}} = \pm 1$ leads to a jump approximately equal to $\pm \theta$ in the remaining agents' belief. As noted in our discussion following Proposition~\ref{prop:bump} we omit the correction $\ep_{T^{(1)}}^+$ to this social information. } 

There are three possible outcomes:  if ${d}^{(1)}_{T^{(1)},0} = 1$ and $y_{T^{(1)},0}^{(i)} \geq 0$, for $i = 2,3$, then both remaining agents decide immediately, ${d}^{(i)}_{T^{(1)},1} = 1$, and the evidence accumulation process stops.  We therefore only examine cases where (i) $y_{T^{(1)},0}^{(i)} < 0$ for $i = 2$ or $i = 3$ (but not both) and (ii) $y_{T^{(1)},0}^{(i)} < 0$ for  both $i = 2,3$. 

Before observing agent 1's decision, and after the private observation at time ${T^{(1)}}$, the beliefs of agents $i = 2,3$ are  $y^{(i)}_{{T^{(1)}},0} = \obb i{1:{T^{(1)}}}$. After observing $d_{T^{(1)},0}^{(1)} = 1$ their beliefs are
\begin{align} \label{E:firstupdate}
y^{(i)}_{{T^{(1)}},1} &= \obb i{1:{T^{(1)}}} + \log \frac{P(d_{T^{(1)},0}^{(1)} = 1, d_{{T^{(1)}}-1}^{(1)} = 0, d_{{T^{(1)}}}^{(2,3)} = 0|\hp)}{P(d_{T^{(1)},0}^{(1)} = 1, d_{{T^{(1)}}-1}^{(1)} = 0, d_{{T^{(1)}}}^{(2,3)} = 0| \hm)} \approx \obb i{1:{T^{(1)}}} + \theta.
\end{align}

Any agent who remains undecided after the update given by Eq.~\eqref{E:firstupdate} updates their belief iteratively in response to the decision information of their neighbor.

\emph{Case (i) - One Agent Undecided.} Without loss of generality, we assume $y_{T,0}^{(2)} \geq 0$ so that  $d_{T,1}^{(2)} = 1$. After observing agent 1's decision, the belief of agent 3 is $y_{T,3}^{(3)} = \ob^{(3)}_{1:T} + \theta$.  A straightforward computation following agent 2's decision gives
\begin{align}
y_{{T^{(1)}},2}^{(3)} &\approx \ob^{(3)}_{1:T^{(1)}} + \theta + \LLR ( {d}_{{T^{(1)}},1}^{(2)} = 1 | {d}_{{T^{(1)}},0}^{(1)} = 1)=  \ob^{(2)}_{1:{T^{(1)}}} + \theta 
+ R_+({T^{(3)}}), \label{E:sum1}
\end{align}
where
\[
R_+(t)  :=  \LLR (y_t^{(2)} \in [0,\theta)| y_{t}^{(2)} \in (-\theta, \theta)) = \LLR (y_t^{(2)} \in [0,\theta))
\]
since agent 2's belief must have been non-negative before observing agent 3's decision. Note that $R_+(t)< \theta$ by the following proposition: 

\begin{proposition}
\label{claim:decbound} 
Let $-\theta < a \leq b < \theta$. If agent $j$ has not made a decision and has accumulated only private information by time $T$, then
\[
\left| \LLR ( a \leq y_T^{(j)} \leq b | y_{t}^{(j)} \in (-\theta, \theta))
 \right| = \left| \LLR ( a \leq y_T^{(j)} \leq b )
 \right| < \theta .
\]
\end{proposition} 

To prove this, we assume an agent's belief lies in some subinterval $(a,b)$ of $\Theta = (-\theta, \theta)$ and bound the social information obtained from this knowledge, concluding it provides at most an increment of social evidence of amplitude $\theta$.  The proof is provided in Appendix~\ref{prop7}. 

Proposition \ref{claim:decbound} implies $R_+(t) < \theta$, but this increment in belief may be sufficient to lead to a decision by agent 1. 
 We will estimate  $R_+(t)$ in arbitrarily large cliques in section~\ref{sec:cliq}. {In particular, $R_+(t)$ may be computed explicitly in the same way as $S_{\pm}(t)$.}
 
\emph{Case (ii): Two Agents Undecided.}
If $y_{{T^{(1)}},0 }^{(i)} < 0$ for  $i = 2,3$, then  both agents remain undecided upon observing agent 3's decision. After each observes this absence of a decision in its counterpart, it follows from the computations that lead to Eq.~(\ref{E:sum1}) that each updates its belief as
\begin{align} \label{E:sum3}
y_{{T^{(1)}}, 2}^{(i)} &\approx \ob^{(i)}_{1:{T^{(1)}}} + \theta + R_-({T^{(1)}}) 
\end{align}
where 
$R_- (t) \equiv   \LLR (   y_{t}^{(\neg i)} \in (-\theta, 0] )$.
Due to symmetry, the new social information is equal for both agents.

Note that $R_- (t) \leq 0$ and also $|R_-| (t) < \theta$, as shown in Proposition~\ref{claim:decbound}.  Therefore Eq.~\eqref{E:sum3} shows that this second belief update cannot lead to a decision, and  $d^{(i)}_{{T^{(1)}},2} = 0$ for $i = 2,3$. {After this second step, the agents cannot obtain further social information and proceed with their private measurements.}

At the end of this exchange, both remaining observers know that the belief of the other is in the non-symmetric interval, $y_{{T^{(1)}},2}^{(\neg i)} \in \big(R_-({T^{(1)}}), \theta +R_-({T^{(1)}})\big)$.   Therefore, future non-decisions become informative, and equilibration follows each private observation as in the case of asymmetric decision thresholds discussed in Section~\ref{sec:bidsymm}.  {One key difference  is that in the present case the undecided agents have been influenced by the decider, and the asymmetry  is determined by the positions of the stochastic belief trajectories at the time of the decision.}

\emph{Concurrent Decisions.} If the first decision is made simultaneously by two agents, the remaining agent receives independent social information from both.  When the two deciding agents disagree, the social information they provide cancels. If the two agents agree on $H^{\pm}$, the undecided agent increases its belief by $\pm 2\theta$ and follows the decision of the other two.  

The exchange of social information increases the probability that all three agents reach a correct decision (Fig.~\ref{3cliquess}a), and decreases the time to a decision, both of a single agent and all agents in the clique (See Fig.~\ref{3cliquess}b). 
This is particularly pronounced when the SNR of private measurements is low.  With highly informative private observations, social information becomes less important.


\begin{figure}[t!]
\centering
\includegraphics[width = 12cm]{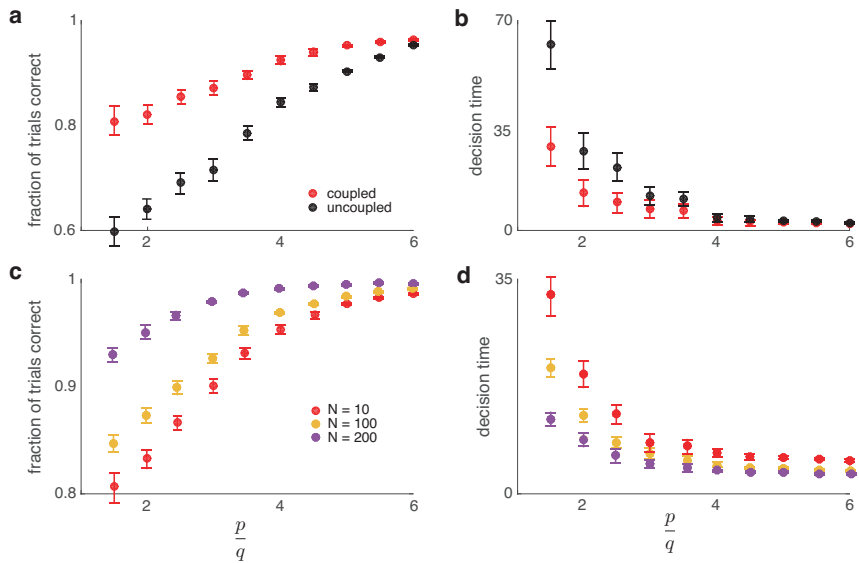}
\caption{(a) In a three-agent clique the probability that all three agents make the correct choice is larger when social information is exchanged. (b) The average time it takes all three agents decide is smaller when the agents communicate, since social information pushes agents closer to threshold. Both effects are more pronounced when the SNR of private evidence is low, \emph{i.e} $p/q$ is small. (c) As the number of agents in a clique, $N$, increases, the probability that all  agents in the clique make the correct choice grows.  The difference is more pronounced when SNR is low.  (d) Larger cliques provide more social information reducing the average decision times. Here, $|\theta_{\pm}| = 30$, and $p,q$ are defined as in Section~\ref{sec:4state}. }
\label{3cliquess}
\end{figure}

\section{Larger Cliques}
\label{sec:cliq}

We  next consider a clique of $N$ agents who all have identical symmetric decision thresholds and measurement distributions.  No social information is exchanged until one or more agents makes a choice at time $T$. We focus on the case of a single agent making the first decision. The case in which two or more agents simultaneously make the first decision can be  analyzed as in the three agent clique, and it often leads to all agents deciding subsequently.  We do not discuss it further.

Without loss of generality, we assume  agent 1 makes the first decision and chooses $\hp$.
This means $y_{T^{(1)},0}^{(1)} \geq \theta$, and thus every other agent, $i$,  updates their belief to 
\[
y_{{T^{(1)}}, 1}^{(i)} = y_{{T^{(1)}},0}^{(j)} + \LLR (d_{{T^{(1)}},0}^{(1)} = 1 ) \approx y_{{T^{(1)}}, 0}^{(i)} + \theta ,  \hspace{5mm} i \neq 1 .
\] 
We assume that the excess information, $\epsilon_{T^{(1)}}^+$ (See Proposition~\ref{prop:bump} for definition), provided by a decision is negligible.

Upon observing a decision, the remaining agents stop making private observations and exchange social information until no further social information is available. Observing $d_{{T^{(1)}},0}^{(1)} = 1$ leads any agent $i$ with belief $y_{{T^{(1)}}, 0}^{(i)} \geq 0$ to  the same decision. We denote the set of these agents by $A_1$ and call these the \defi{agreeing agents}.  We will see that there can be multiple waves of agreeing agents. The agents whose beliefs satisfy $y_{{T^{(1)}}, 0}^{(i)} < 0$ update their belief to $0 < y_{{T^{(1)}}, 1}^{(i)} < \theta$ but do not make a decision.  We denote the set of these agents by $U_1$ and call these \defi{undecided agents}.  

The decisions of agreeing agents are conditionally independent given the observed decision of agent 1. Thus, each agreeing agent independently provides additional evidence for $\hp,$ while each undecided agent provides evidence for $\hm$. As in the case of three agents, the social information provided by an agreeing agent is $R_+({T^{(1)}})$ and for an undecided agent  is $R_-({T^{(1)}}) = - R_+({T^{(1)}}),$ where $R_+({T^{(1)}})$ is given in Eq.~(\ref{posprob}). The equality follows from our assumption of symmetry, as we now show.

\begin{proposition}
Assume agent $i$ has not received any social information at time $t$, so their belief, $y_t^{(i)},$ is based on private information. Also assume that the decision thresholds, $\theta_{\pm} = \pm \theta,$ and measurement distributions are symmetric. Let
\begin{align*}
R_{+}(t) : =   \LLR (y_t^{(i)} \in [0, \theta)) \qquad \text{ and } \qquad
 R_{-}(t) : =  \LLR (y_t^{(i)} \in (-\theta,0] ).
\end{align*}
Then $R_-(t) = -R_+(t)$.
\end{proposition}

\begin{proof}
Following from the argument in Proposition \ref{claim:decbound}, we can compute
\begin{align*}
P(y_{t}^{(i)} \in [0,\infty) | H^{\pm}) = \sum_{V_T^+} \prod_{t=1}^T f_{\pm}( \xi_t^{(i)}),
\end{align*}
where $V_T^+ =  \{ \xi_{1:T}^{(j)} \in \Xi^T 
: \obb j {1:s}\in (-\theta, \theta ), \text{ for } s \leq T, \; \obb j {1:T}\in [0,\infty) \}
. $
By symmetry, we know that for any $\xi_{1:t}^{(i)} \in V_T^+$, there exists $- \xi_{1:t}^{(i)} \in V_T^-$ where
$
V_T^- =  \{ \xi_{1:T}^{(j)} \in \Xi^T 
: \obb j {1:s}\in (-\theta, \theta ), \text{ for } s \leq T, \; \obb j {1:T}\in (-\infty,0] \}
 $
and vice versa. Therefore we can write
\begin{align*}
P(y_{t}^{(i)} \in (-\infty, 0] | H^{\pm}) &=  \sum_{V_T^-} \prod_{t=1}^T f_{\pm}( \xi_t^{(i)}) = \sum_{V_T^+}   \prod_{t=1}^T f_{\pm}(-\xi_t^{(i)}) = \sum_{V_T^+}  \prod_{t=1}^T f_{\mp}(\xi_t^{(i)}) = P(y_{t}^{(i)} \in [0,\infty) | H^{\mp}),
\end{align*}
where the penultimate equality holds since $f_+(\xi) = f_-(- \xi)$. Therefore,
\begin{align*}
R_-(t) &=  \LLR (y_t^{(i)} \in (-\theta,0] ) = \log \frac{P(y_t^{(i)} \in [0, \theta) | H^-)}{P(y_t^{(i)} \in [0, \theta) | H^+)} = - \LLR (y_t^{(i)} \in [0, \theta) ) = -R_+(t).
\end{align*}
\end{proof}

Note that $N = a_1 + u_1 + 1$, where $a_1$ and $u_1$ are the number of agreeing and undecided agents in the first wave following the decision of agent 1 at time  $T^{(1)}$.  {To compute the evidence update for undecided agent $j$ after the first wave, note that
\begin{align*}
P(d_{{T^{(1)}},0}^{(1)} = 1, d_{{T^{(1)}}, 1}^{(i)} , d_{{T^{(1)}}, 1}^{(j)} |  H) & = P( d_{{T^{(1)}}, 1}^{(i)} , d_{{T^{(1)}}, 1}^{(j)} |d_{{T^{(1)}},0}^{(1)} = 1,  H)  P(d_{{T^{(1)}},0}^{(1)} = 1 | H) \\
& = P( d_{{T^{(1)}}, 1}^{(i)}  |d_{{T^{(1)}},0}^{(1)} = 1,  H) P(d_{{T^{(1)}}, 1}^{(j)} |d_{{T^{(1)}},0}^{(1)} = 1,  H)  P(d_{{T^{(1)}},0}^{(1)} = 1 | H)
\end{align*}
for any pair of agents $i \neq j$ different from 1. Therefore, after observing this first wave of decision, all remaining undecided agents update their belief as
\begin{align*}
y_{{T^{(1)}}, 2}^{(j)} = y_{{T^{(1)}}}^{(j)} & +  \LLR (d_{{T^{(1)}},0}^{(1)} = 1)  + \sum_{i \in A_1} \LLR (d_{{T^{(1)}}, 1}^{(i)} = 1 | d_{{T^{(1)}},0}^{(1)} = 1) \\
& + \sum_{k \in U_1-\{j\}} \LLR (d_{{T^{(1)}}, 1}^{(k)} = 0 | d_{{T^{(1)}},0}^{(1)} = 1) .
\end{align*}}
By conditional independence, this simplifies to
\begin{align*}
y_{{T^{(1)}}, 2}^{(j)} &= y_{{T^{(1)}},0}^{(j)} + \theta + a_1  \LLR (0 \leq y_{{T^{(1)}},0}^{(i)} < \theta) + (u_1-1) \LLR ( -\theta < y_{{T^{(1)}},0}^{(k)} < 0)\\
&= y_{{T^{(1)}},0}^{(j)} + \theta + a_1  R_+ + (u_1-1)R_-,
\end{align*}
where $R_{\pm}$ are defined as in Eqs.~(\ref{E:sum3}) and ~(\ref{posprob}).
By symmetry, $R_-({T^{(1)}}) = -R_+({T^{(1)}})$.  Thus, all undecided agents update their belief as
\begin{equation}
y_{T,2}^{(j)} =y_{T,0}^{(j)} + \theta + (a_1 - (u_1 -1)) R_+({T^{(1)}}), \qquad j \in U_1.
\label{cliquejump}
\end{equation}
{Note that each undecided agent observes all $a_1$ agreeing agents and all other  $u_1-1$ undecided agents.}

 We will see that multiple rounds of social evidence exchange can follow.  Let 
$$
R\big(t,(a,b)\big) : =  \LLR (y_t^{(j)} \in (a,b)| y_{t}^{(j)} \in (-\theta, \theta)) =  \LLR (y_t^{(j)} \in (a,b)). 
$$
Note that $|R\big(t,(a,b)\big)|  \leq \theta$ by Proposition~\ref{claim:decbound} and that $R\big(t,[0,\theta)\big) = R_+(t)$, $R\big(t,(-\theta,0]\big) = R_-(t)$.

Each remaining undecided agent has now observed the decision of the $a_1$ agreeing agents and the indecision of the other $u_1-1$ undecided agents, excluding itself.  Eq.~(\ref{cliquejump}) implies several possibilities for what these agents do next: \\[0ex]
\begin{enumerate}
\item If the number of agreeing agents, $a_1$, exceeds the number of undecided agents, $u_1$, and satisfies $(a_1-(u_1-1)) R_+ ({T^{(1)}})\geq \theta$, all the remaining undecided agents, $j \in U_1,$ go along with the decision of the first agent and  choose $\hp$. The second wave of agreeing agents encompasses the remainder of the network.
\item If the number of undecided agents, $u_1$, exceeds the number of agreeing agents, $a_1$, by a sufficient amount, $(a_1 -(u_1-1)) R_+({T^{(1)}}) \leq - 2 \theta$, then  all undecided agents update their belief to  $y_{{T^{(1)}},2}^{(j)} \leq - \theta,$ $j \in U_1$. This is a somewhat counterintuitive situation:  if sufficiently many agents remain undecided after observing the first agents' choice, then, after observing each other's absence of a decision, they all agree on the opposite.  A wave of agreement is followed by a larger wave of disagreement.
\item If $- 2\theta < (a_1 -(u_1-1)) R_+({T^{(1)}}) < -\theta$, then some of the remaining agents may disagree with the original decision and choose $\hm$, while others may remain undecided. We call the set of disagreeing (contrary) agents $C_2$ and the set of still undecided agents $U_2$.  We denote the sizes of these sets by $c_2$ and $u_2$, respectively. Note that $U_1 = C_2 \cup U_2$. 

The  agents in $U_2$ know that  $y_{{T^{(1)}},0}^{(j)} \in ( - \theta, ((u_1-1)-a_1) R_+({T^{(1)}}) - 2\theta]$ for all $j \in C_2$ and   $y_{{T^{(1)}},0}^{(j)} \in (((u_1-1)-a_1) R_+({T^{(1)}}) - 2\theta,0)$ for all $j \in U_2$.  All agents in $U_2$ thus update their belief again to 
\begin{equation} \label{E:refined}
\begin{split}
y_{{T^{(1)}},2}^{(j)} &= y_{{T^{(1)}},0}^{(j)} + \theta + a_1 R_+({T^{(1)}}) + (u_2-1) R\big({T^{(1)}}, (((u_1-1)-a_1) R_+({T^{(1)}}) - 2\theta,0)\big) \\
                    &+ c_2 R\big({T^{(1)}},(- \theta, ((u_1-1)-a_1) R_+({T^{(1)}}) - 2\theta)\big).
\end{split}
\end{equation}
This update includes the social information obtained from the initial observation, $\theta$, and from the agreeing agents in the first round, $a_1 R_+({T^{(1)}})$. The last two summands in Eq.~\eqref{E:refined} give a refinement of the social information from the originally undecided agents in $U_1$.  As a result  some undecided agents in $U_2$ can make a choice.  If so, the process repeats until no undecided agents are left or no decisions occur after an update. This process thus must terminate after a finite number of steps.  This is akin to the  equilibration of social information described earlier, but it involves observed choices and occurs across the network. However,  the process depends on the private evidence accumulated by the undecided agents and thus depends on realization. 
 \item If $- \theta \leq (a_1 - (u_1 -1)) R_+({T^{(1)}}) \leq 0$, then no agents in $U_1$ make a decision, and no undecided agent obtains any further social information. They thus continue to gather private information. Symmetry is broken as they know that $y_{{T^{(1)}},0}^{(j)} \in (- \theta,0]$ for all remaining agents, $j \in U_1$. 
\item If $ 0 < (a_1 -(u_1-1)) R_+({T^{(1)}}) < \theta$, some undecided agents may choose $\hp$, and some may remain undecided. We call the first set $A_2$ and the second $U_2$, and we denote the sets' cardinality by $a_2$ and $u_2$, respectively.  In this case, $U_1 = A_2 \cup U_2$.  All undecided agents $j \in U_2$ know that $y_{{T^{(1)}},0}^{(j)} \in [ -(a_1-(u_1-1)) R_+({T^{(1)}}),0)$ for all $j \in A_2$, and  $y_{{T^{(1)}},0}^{(j)} \in (- \theta, -(a_1-(u_1-1)) R_+({T^{(1)}}))$ for all $j \in U_2$. They thus update their belief to
\begin{equation} \label{E:refined2}
\begin{split}
y_{{T^{(1)}},2}^{(j)} &= y_{T,0}^{(j)} + \theta + a_1 R_+({T^{(1)}}) + (u_2-1) R\big({T^{(1)}},(- \theta, -(a_1-(u_1-1)) R_+({T^{(1)}}))\big) \\
 & + a_2 R\big({T^{(1)}},[ -(a_1-(u_1-1)) R_+({T^{(1)}}),0)\big).
 \end{split}
\end{equation}
If some agents decide as a result, the process continues as discussed in 3. \\
\end{enumerate}

\begin{figure}
\centering
\includegraphics[width = 13cm]{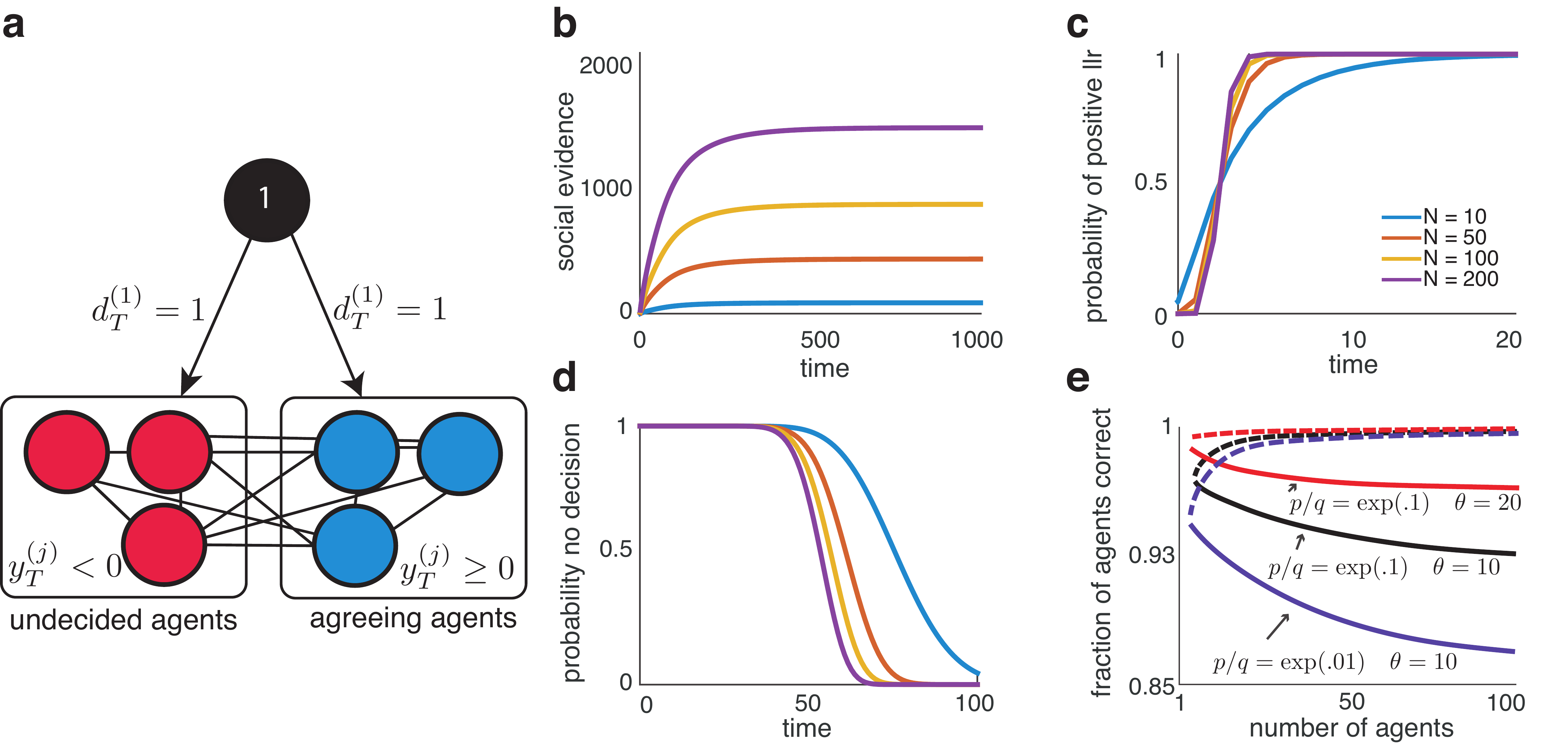}
\caption{ (a) In a clique of size $N$, agent 1 makes a decision which is observed by all other agents.  Agents whose beliefs are positive (agreeing agents),  follow the decision of this initial agent.  The remaining undecided agents continue to exchange social information.  (b)  Average social information, $E[(a_1-(u_1 - 1))R_+],$ available to undecided agents after observing the first wave of agreeing agents as a function of the first decision time for different clique sizes.  (c) The probability that the majority of the group has positive belief as a function of the first decision time for different clique sizes.  This is equivalent to the probability that the majority of the clique consists of agreeing agents when the first decision occurs.  (d) The probability that no agent has made a decision as a function of decision time for different clique sizes.  Here, $|\theta_{\pm}| = 20$, $p= e/10$, and $q = 1/10$. (e) The fraction of agents choosing correctly after the first wave following a decision (solid) and after equilibration (dashed) as a function of clique size for various values of $p/q$ and $\theta$.}
\label{cliques}
\end{figure}


To summarize, before any agent in a clique makes a decision, each agent accumulates private evidence independently.  A given agent will make the first decision, say $\hp$. All other agents with positive evidence will follow this choice in a first wave of agreement.  Undecided agents modify their evidence according to Eq.~(\ref{cliquejump}).  How many agents agree $a_n$, disagree $c_n$, or remain undecided $u_n$ in exchange $n$ depends on the numbers of agreeing, disagreeing, and undecided agents from previous steps: $a_{1:n-1}$, $c_{1:n-1}$, and $u_{1:n-1}$.

\subsection*{Large System Size Limit}

In Fig.~\ref{cliques} we plot information about the behavior of cliques of various sizes, hinting at trends that emerge as $N \to \infty$.  As the size of the clique grows, so does the amount of information available to undecided agents after the first wave of decisions, {\em{i.e.}} the members of the set $U_1$ (See Fig.~\ref{cliques}b).  As clique size grows, the first decision occurs earlier (Figs.~\ref{cliques}c-d).  In particular, as $N$ grows, the first decision time approaches $\gamma \equiv \theta/\log(p/q)$--the minimum number of steps required to make a decision.  Moreover,  most agents accrue evidence in favor of the correct decision at an earlier time.  By the time the first decision occurs, the majority of the clique will be inclined toward the correct decision.  The ensuing first decision, if it is correct, will immediately cause the majority of the clique to choose correctly.  However, note that as clique size grows, the probability that the initial decision is incorrect also grows.  It is therefore not obvious that the asymptotic fraction of the clique that chooses correctly approaches 1:  if the initial decision is incorrect, the majority of the clique will be undecided after incorporating social information from the initial decision (See Fig.~\ref{cliques}c).  

However, the social information exchange described above can still lead the remainder of the clique to overcome the impact of the wrong initial decision.    {To see this, note that as $N \to \infty$, a fraction $p^\gamma$ of the agents will choose $\hp$ and another fraction $q^\gamma$ of agents will choose $\hm$ simultaneously at time $\gamma$.  As all belief trajectories are independent until this time, all undecided agents observing these decisions will  increase their belief by $\Gamma \approx (p^\gamma - q^\gamma)N\theta$, with the approximation becoming exact in the limit of large $N$.  Since $p^{\gamma} > q^{\gamma}$, for sufficiently large $N$, with high probability $\Gamma > 2\theta$, implying that all undecided agents at time $\gamma$ choose $\hp$.}

\section{Discussion}
\label{sec:discussion}

Our evidence accumulation model shows how to best combine streams of private and social evidence to make decisions. There has been extensive work on the mathematics of evidence accumulation by individual agents~\cite{Wald1948,Busemeyer1993,Bogacz2006,veliz16,radillo17}. The resulting models have been very successful in explaining experimental data and the dynamics of decision making~\cite{Usher2001,Gold2007}, and there is evidence that they can explain the formation of social decisions~\cite{Krajbich2015}.  In parallel, economists have developed mathematical models of networks of rational decision makers~\cite{Goyal2012,mueller2013general}.  However,  economists predominantly consider static models as the norm, and  the temporal dynamics of decisions are not studied explicitly. In contrast, temporal dynamics have been extensively studied in the psychology literature but predominantly in the case of subjects making decisions based on private evidence~\cite{Luce86} (although see~\cite{Bahrami2010,Krajbich2015,haller2018}).  Our model thus 
  provides a bridge between approaches and ideas in several different fields.
  
   The  belief dynamics of agents we considered differs from past accounts in several respects. First, we assume agents only communicate decision information. In this case, coupling between agents is nonlinear and depends on a threshold-crossing process. We have shown that agents can exchange social evidence  even in the absence of a decision when their measurement distributions or decision thresholds are not symmetric. Observing the decision state of a neighbor may result in further exchanges of social information with  other neighbors.    This can lead to complex computations, in contrast to the model discussed in~\cite{kimura2012group}, where a central agent sees all actions and optimizes a final decision assuming all other agents accumulate evidence independently.

 The beliefs  of evidence-accumulators in a network have been modeled by diffusively coupled drift-diffusion equations~\cite{Srivastava2014, poulakakis2012node}. Such models approximate continuous information-sharing between agents~\cite{Srivastava2014}.  However,  it is not always clear when linear coupling between neighbors approximates the normative process of evidence accumulation and exchange between rational agents. 
In  a related model,  an agent observing a choice was assumed to experience a jump in its belief~\cite{caginalp2017decision}. This study examined how the size of this jump affects the decisions of the collective. While 
not normative, this model  provides insight into decision making when agents over- or under-weight available evidence~\cite{Beck:2012}.

The absence of choice is informative only in the case of asymmetries.  This could be due to one choice being more rewarding than another~\cite{mulder2012,hanks2011,voss2004}.
For instance, the random dot motion task, which requires subjects  to determine the dominant direction of randomly moving dots, has been extensively used in psychophysical and neurobiological studies. In this task subjects set their decision thresholds symmetrically when the reward and frequency of each choice is symmetric~\cite{shadlen1996,Gold2007}. However, when one drift direction is more likely or rewarded more, subjects exhibit accuracy and response time biases consistent with an asymmetric prior or thresholds~\cite{mulder2012,hanks2011}.
Recordings suggest neural activity representing decisions is dynamically modulated during the evidence accumulation process to reflect this bias~\cite{hanks2011,platt1999}. This is notably different from the static offsets apparent in the Bayesian model, but further investigation is needed to determine how priors are represented across the entire decision-making system~\cite{summerfield2014}. 

In our model, unless threshold asymmetries are large, the social information communicated by the indecision of a single agent is weak. However, in large social networks the totality of such evidence provided by all other agents can be large compared to the private information obtained by an agent. Therefore, in large networks even small asymmetries can produce social information from indecisions that strongly influences decisions.  Consider a new, desirable product that hits the market. If a large fraction of a social network is indecisive about purchasing the product, this can communicate potential issues with the product. This signal could be particularly strong if the product is  particularly desirable upon its release.  

We assumed that all agents implement the same decision thresholds $\theta_{\pm}$ and that they all know this. 
This assumption can be relaxed in different ways:  agents could use different decision thresholds $\theta_{\pm}^{1:N},$ which may or may not be known to the other agents in the network.  When thresholds are known, much of  our analysis carries over, with agents updating their belief in response to  observing a decision state depending on who they observe.  This increases the notational complexity, but the opinion exchange dynamics
remain largely the same. When thresholds are not known, agents can also update their belief about, \emph{i.e.} the posterior over, the possible decision thresholds of their neighbors over multiple decisions.
Indeed, finding out which of your friends are rash to make decisions and which ones tend to be cautious is essential to properly weighing their opinions. However, this requires observations and interactions over multiple decisions--something we do not consider here. 
In addition, there is evidence that humans tend to treat others as if they make decisions in similar ways~\cite{Bahrami2010}. 
How humans learn each other's decision thresholds and how such learning  impacts decision making has not been extensively explored and provides 
fruitful avenues for  future work.

{
We have assumed that each observer maintains a constant, predetermined threshold while accumulating evidence. In many situations this is likely not the best policy.  We could extend our modeling framework to thresholds that can change in response to accumulated private and social evidence as well as time pressure. Such approaches have been developed to describe single ideal observers accumulating evidence in uncertain environments and rely on dynamic programming to define the agent's decision policy~\cite{drugowitsch12}. This can result in decision thresholds that vary in time and that can be realization-dependent~\cite{drugowitsch14}.  In a social network, to optimize their own threshold, every agent needs to model the process by which other agents choose their threshold. This can lead to the need for recursive reasoning to set the thresholds appropriately. In practice one can use the setting of partially observed Markov processes to implement such computations~\cite{khalvati2019}.
}

Our Bayesian model is unlikely to accurately describe the decision making process of biological agents. The rationality of humans and other animals is bounded~\cite{simon1990}, while some of the computations that we have described are quite complex and provide only marginal increases in the probability of making a correct choice.   Thus biological agents are likely to perform only those computations -- perhaps approximately -- that provide the largest impact~\cite{Couzin2009}.  A normative model allows us to identify which computations are important and which offer only fractional benefits. 


Our analysis  relied on the assumption that private observations received by the individual agents are independent. When this is not the case agents need to account for correlations between private observations in determining the value of social information. In general, we expect that an increased redundancy of neighboring agents' observations  reduces the impact of social information. However, experimental evidence suggests that humans do not take into account such correlations when making decisions~\cite{enke2013correlation,levy2015correlation} and behave as if the evidence gathered by different members of their social group  is independent.  {Indeed, it can be shown that properly integrating social information from unseen agents can be NP-hard, although the computations can be tractable in some networks~\cite{hkazla2017bayesian,hazla2018reasoning}}.  Our model may approximate the decision-making process of agents that make such assumptions. The impact of ignoring such correlations when they are present is an interesting avenue for future work~\cite{stolarczyk17}.

Our findings are distinct from but related to previous work on herding and common knowledge. In this modeling work agents typically make a single private measurement, followed by an exchange, or propagation of social information~\cite{banerjee1992,milgrom1982}. In recurrent networks, agents can announce their preference for a choice, until they reach agreement with their neighbors~\cite{aumann1976,Geanakoplos1982,geanakoplos1994,mossel2014opinion}. This framework can be simplified so that agents make binary decisions, based solely on their neighbors' opinion, admitting asymptotic analyses of the cascade of decisions through networks with complex structure~\cite{Watts2002}. In contrast, we assume that agents accumulate private and social evidence and make irrevocable decisions. {The resulting decision-making processes are considerably different: for instance, in our case there is no guarantee that social learning occurs. Combining private and social evidence also makes it difficult to derive exact results, but we expect asymptotic formulas are possible for large cliques with simpler assumptions on the decision process.}


%
%

\section*{Acknowledgements}
This work was supported by an NSF/NIH CRCNS grant (R01MH115557) and an NSF grant (DMS-1517629). BK and KJ were supported by NSF grant (DMS-1662305). KJ was also supported by NSF NeuroNex grant (DBI-1707400). ZPK was also supported by an NSF grant (DMS-1615737).

\appendix 

\section{Bounding social information from a decision in a recurrent two-agent network}
\label{bnd2coup}
Here we prove Proposition~\ref{prop:bump2} using an argument similar to that  in Proposition~\ref{prop:bump}.  First, consider the case in which $n = 0$, so the private observation $\xi_{T^{(i)}}^{(i)}$ triggers the decision. If $d_{T,0}^{(i)} = + 1$, we know
\begin{align*}
y_{T,0}^{(i)} = \obb{i}{1:T} + \soc_{T-1} \geq \thp ,
\end{align*}
and since $ \obb{i} {1:T-1} + \soc_{T-1}  < \thp$, then
\begin{align*}
y_{T^{(i)},0}^{(i)} = \obb{i} {1:T} + \soc_{T-1} < \thp + \LLR (\xi_{T}^{(i)}).
\end{align*}
Marginalizing over all chains ${\mc C}_{+}(T,0)$ such that $y_{s,m}^{(i)} \in \Theta$ for all $0 < s < T$, and corresponding $0 \leq m \leq N_s$, preceding the decision $d_{T,0}^{(i)} = + 1$ at $(T,0)$, we thus find
\begin{align*}
\theta_+ - \soc_{T-1,0} \leq \LLR (d_{T,0}^{(j)} = 0, d_{T-1,N_{t-1}}^{(i)} = 0, d_{T,0}^{(i)} = 1)< \theta_+ - \soc_{T-1,0} + \ep_{T,0}^+,
\end{align*}
where $N_{T-1}$ is the maximal substep in the social information exchange following the private observation at timestep $T-1$.

On the other hand, suppose $d_{T,n}^{(i)} = + 1$, and $0 < n \leq N_{T}$ so that the decision is reached during the social information exchange following an observation. Then,
\begin{align*}
\obb{i}{1:T} + \soc_{T-1,n-1}^{(i)} < \theta_+,
\end{align*}
which implies
\begin{align*}
\theta_+ \leq y_{T,n}^{(i)} = \obb{i}{1:T} + \soc_{T-1,n-1} +(\soc_{T,n} -  \soc_{T,n-1})   < \theta_+ +  (\soc_{T,n} -  \soc_{T,n-1}).
\end{align*}
Marginalizing over all chains ${\mc C}_+(T,n)$ such that $y_{s,m}^{(i)} \in \Theta$  for all $0 < s \leq T$, and corresponding $0 \leq m \leq N_s$, preceding the decision at $(T,n)$,
\begin{align*}
\theta_+ - \soc_{T,n} \leq \LLR ( d_{T^{(i)},n}^{(j)} = 0, d_{T^{(i)},n-1}^{(i)} = 0, d_{T^{(i)},n}^{(i)} = 1) < \theta_+ - \soc_{T,n} + \ep_{T,n}^+,
\end{align*}
where
\begin{align*}
\ep_{T,n}^+
\leq (\soc_{T,n} -  \soc_{T,n-1}).
\end{align*}
Following the arguments in Theorem 5.1, we note then that
\begin{align*}
y_{T,n+1}^{(j)} & = \LLR (\xi_{1:T}^{(j)}, d_{1:T,0:n}^{(j)} = 0, d_{1:T,0:n-1}^{(i)} = 0, d_{T,n}^{(i)} = \pm 1) \\
& = \LLR (\xi_{1:T}^{(j)}) + \LLR ( d_{T,n}^{(j)} = 0, d_{T,n-1}^{(i)} = 0, d_{T,n}^{(i)} = -1) \\
& =   \obb {j}{1:T} + \soc_{T,n+1}^+.
\end{align*}
The proof follows similarly for the case $d_{T,n}^{(i)} = - 1$.





\section{Proof of Proposition~\ref{claim:decbound}}
\label{prop7}  To simplify the proof we assume that a sequence of private observations results in a belief (LLR)
that can lie exactly on the decision threshold $\theta$, as in the examples in the text. 
Define the subset $V_T(a,b) \subset \Xi^T$ of the product space of observations $\Xi^T$ consisting of  observation sequences that result in a belief contained in $[a,b]$ at time $T$ :
\[
V_T(a,b) = 
\{ \xi_{1:T}^{(j)} \in \Xi^T 
: \obb j {1:t}\in (-\theta, \theta ), \text{ for } t \leq T, \; \obb j {1:T}\in [a,b] \}
. \]
By definition, we can write
\[
e^a \leq \prod_{t = 1}^{T} \frac{f_+ ( \xi_t^{(j)})}{f_-( \xi_t^{(j)})} \leq e^b, \hspace{5mm} \forall \xi_{1:T} \in V_T(a,b),
\]
which can be rearranged as
\[
e^a \prod_{t=1}^T f_-( \xi_t^{(j)}) \leq \prod_{t = 1}^{T} f_+ ( \xi_t^{(j)}) \leq e^b  \prod_{t=1}^T f_-( \xi_t^{(j)}), \hspace{5mm} \forall \xi_{1:T}^{(j)} \in V_T(a,b).
\]
Summing over all $\xi_{1:T}^{(j)} \in V_T(a,b)$ then yields
\begin{align*}
e^a \sum_{V_T(a,b)} \prod_{t=1}^T f_-( \xi_t^{(j)}) \leq \sum_{V_T(a,b)} \prod_{t = 1}^{T} P_+( \xi_t^{(j)}) \leq e^b  \sum_{V_T(a,b)} \prod_{t=1}^T P_-( \xi_t^{(j)}),
\end{align*}
so that by noting $P(y_{T}^{(j)} \in [a,b] | H^{\pm}) = \sum_{V_T(a,b)} \prod_{t=1}^T P_{\pm}( \xi_t^{(j)})$, and rearranging we find
\begin{align*}
e^{- \theta} < e^a \leq \frac{P(y_{T}^{(j)} \in [a,b] | H^{+})}{P(y_{T}^{(j)} \in [a,b] | H^{-})} \leq e^b < e^{\theta},
\end{align*}
which implies the desired result.

\bibliography{ms}
\bibliographystyle{siam}
\end{document}